\documentclass[aps,prx,reprint,showpacs,twocolumn,superscriptaddress,eqsecnum]{revtex4-2}
\usepackage{xr}
\usepackage{amsmath}
\usepackage{amsfonts,amssymb,amsthm,amsxtra}
\usepackage{algorithm} 
\usepackage{algpseudocode} 
\usepackage[]{graphicx}
\usepackage{mathrsfs}
\usepackage{bbm}
\usepackage{bm}
\usepackage{braket}
\usepackage{xcolor}
\usepackage{mathtools}
\usepackage[bookmarks=false,colorlinks=true,urlcolor=blue,citecolor=blue,linkcolor=blue]{hyperref}
\usepackage[normalem]{ulem}

\newtheorem{theorem}{Theorem}
\newtheorem*{theorem*}{Theorem}
\newtheorem{lemma}[theorem]{Lemma}
\newtheorem{claim}[theorem]{Claim}

\newtheorem{definition}[theorem]{Definition}

\newcommand{\Z}{\mathbb{Z}}

\DeclareMathOperator{\depth}{depth}
\DeclareMathOperator{\ord}{ord}

\DeclarePairedDelimiterXPP\bigO[1]{O}{(}{)}{}{#1}
\DeclarePairedDelimiterXPP\bigomega[1]{$\Omega$}{(}{)}{}{#1}
\DeclarePairedDelimiterXPP\bigtheta[1]{$\Theta$}{(}{)}{}{#1}

\begin{document}

\title{Toward a 2D Local Implementation of Quantum LDPC Codes}

\author{Noah Berthusen}
\email{nfbert@umd.edu}
\author{Dhruv Devulapalli}
\affiliation{Joint Center for Quantum Information and Computer Science, NIST/University of Maryland, College Park, Maryland 20742, USA}
\author{Eddie Schoute}
\affiliation{Computer, Computational, and Statistical Sciences Division, Los Alamos National Laboratory, Los Alamos, NM 87545, USA}
\author{Andrew M. Childs}
\affiliation{Joint Center for Quantum Information and Computer Science, NIST/University of Maryland, College Park, Maryland 20742, USA}
\affiliation{Department of Computer Science and Institute for Advanced Computer Studies,
University of Maryland, College Park, MD 20742, USA}
\author{Michael J. Gullans}
\affiliation{Joint Center for Quantum Information and Computer Science, NIST/University of Maryland, College Park, Maryland 20742, USA}
\author{Alexey V. Gorshkov}
\affiliation{Joint Center for Quantum Information and Computer Science, NIST/University of Maryland, College Park, Maryland 20742, USA}
\affiliation{Joint Quantum Institute, NIST/University of Maryland, College Park, Maryland 20742, USA}
\author{Daniel Gottesman}
\affiliation{Joint Center for Quantum Information and Computer Science, NIST/University of Maryland, College Park, Maryland 20742, USA}
\affiliation{Department of Computer Science and Institute for Advanced Computer Studies,
University of Maryland, College Park, MD 20742, USA}

\begin{abstract}
Geometric locality is an important theoretical and practical factor for quantum low-density parity-check (qLDPC) codes which affects code performance and ease of physical realization.
For device architectures restricted to 2D local gates, naively implementing the high-rate codes suitable for low-overhead fault-tolerant quantum computing incurs prohibitive overhead.
In this work, we present an error correction protocol built on a bilayer architecture that aims to reduce operational overheads when restricted to 2D local gates by measuring some generators less frequently than others. 
We investigate the family of bivariate bicycle qLDPC codes and show that they are well suited for a parallel syndrome measurement scheme using fast routing with local operations and classical communication (LOCC).
Through circuit-level simulations, we find that in some parameter regimes bivariate bicycle codes implemented with this protocol have logical error rates comparable to the surface code while using fewer physical qubits.
\end{abstract}


\maketitle


\section{Introduction}

The surface code, despite showing promising theoretical and experimental performance~\cite{bravyi1998quantum,Kitaev_2003,Fowler2009,Fowler2012,Zhao2022,google2023suppressing}, is poorly suited to large-scale fault-tolerant quantum computation due to its large qubit overhead~\cite{Fowler2012,Litinski_2019,beverland2022assessing}. As a result, there has been much effort on the development of high-rate quantum low-density parity-check (qLDPC) codes~\cite{Breuckmann_2021}. 
As these codes can encode multiple logical qubits, the required space resources are reduced, in some instances, to a constant~\cite{gottesman2014faulttolerant}.

One of the main drawbacks of these high-rate qLDPC codes is that many long-range connections are needed to implement their syndrome extraction circuits~\cite{Bravyi_Terhal_2009, Bravyi_Poulin_Terhal_2010, Baspin_2022, Baspin_2022_2}.
This is a pressing issue for architectures such as superconducting qubits. There, many of the current designs only allow two-qubit gates to be performed between qubits that are 2D nearest neighbors, in which case implementing these long-range entangling gates incurs significant overhead~\cite{Delfosse_Beverland_Tremblay_2021, baspin2023lower}. 
Several recent proposals have attempted to alleviate this overhead by taking advantage of more complex electrical wiring of the superconducting circuits~\cite{Tremblay_2022_thin, bravyi2023highthreshold}, employing code concatenation~\cite{pattison2023hierarchical, gidney2023yoked}, or using bosonic cat qubits~\cite{ruiz2024ldpccat}.
Implementing these long-range connections is less problematic in architectures like neutral atoms, ion traps, or semiconductor spin qubits that can implement long-range gates through qubit movement~\cite{ryananderson2022implementing, xu2023constantoverhead, viszlai2023matching, Bluvstein_2023, moses2023,Burkard23, desmet2024highfidelitysinglespinshuttlingsilicon}.
However, since movement adds additional complications associated with qubit decoherence, heating, and loss, it is worthwhile to consider schemes that limit the amount of movement. In the extreme case, one can consider qubits that are fixed in space and solely use local interactions to perform entangling gates. Such studies provide additional insight into the tradeoffs associated with engineering long-range connectivity through qubit motion or more complex electrical wiring.

In this paper, we present an approach to qLDPC codes that works without qubit motion or long-range couplers, inspired by the so-called stacked model~\cite{Baspin_2022, berthusen2023}. In this model, we assume that the high-rate qLDPC codes of interest have the property that after embedding the code into $\mathbb{Z}^2$, the majority of the stabilizer generators are local; that is, their qubits are contained within a ball of constant radius.
We claim that most of the work required to perform the syndrome extraction circuit with 2D local gates comes from measuring the few nonlocal generators,
so measuring these generators less frequently has the potential to significantly reduce the time overhead, ideally at only a minor cost to the error correction performance of the code.
It was shown in Ref.~\cite{berthusen2023} that, for quantum expander codes~\cite{Leverrier_2015}, neglecting to measure a large percentage of generators could be reasonably tolerated; however, the model considered there was narrow in scope, considering only a phenomenological noise model and neglecting the problem of embedding the codes.
It is therefore unclear whether such codes lend themselves well to physical implementations. 
Nonetheless, their results on partial error correction were optimistic and motivated the investigation of the more realistic architecture developed in this work.

We propose and benchmark a realistic bilayer architecture suited for near- to mid-term superconducting devices and other platforms with restricted qubit movement. 
We find that the recently introduced bivariate bicycle (BB) qLDPC codes~\cite{bravyi2023highthreshold} coming from the larger family of generalized bicycle qLDPC codes~\cite{Kovalev2013} are well suited for both the stacked model and the bilayer architecture.
These codes have natural embeddings into $\mathbb{Z}^2$ where the generators have a repeated structure, and in some instances, a majority of the generators are geometrically small.
The first property makes them amenable to a parallel syndrome measurement scheme using routing with fast local operations and classical communication (LOCC), and the second property makes them good candidates for reducing overhead using the stacked model.
More generally, we develop bounds on how quickly syndrome extraction can be performed in this manner and provide an algorithm to do so.
Overall, we find that over multiple rounds of decoding, BB codes implemented in this architecture have error correction performance comparable to the standard (rotated) surface code, albeit only when the parameters in the error model lie in certain regimes. 

Our work stands as an alternative architecture that may be more practical for near-term quantum computers without the ability to move qubits. As such, it is incomparable to schemes such as Ref.~\cite{xu2023constantoverhead, viszlai2023matching} which allow for qubit movement.
Several recent works have also proposed layered architectures~\cite{Tremblay_2022_thin, bravyi2023highthreshold}; however, their motivation is in minimizing the number of crossings in the two-qubit gate connectivity. They achieve this through the use of long-range connections, the elimination of which is the main imposed constraint of our work. 
Refs.~\cite{pattison2023hierarchical, gidney2023yoked} present asymptotically well performing concatenated schemes which use only local connectivity; however, the required overheads likely make them infeasible for near- and mid-term quantum computers. In particular, Ref.~\cite{gidney2023yoked} estimates that $\sim$600 physical qubits would be needed per each logical qubit, which is an order of magnitude more than what our architecture needs to implement the $[[144,12,12]]$ Gross code~\cite{bravyi2023highthreshold}, with $\sim$48 physical qubits per logical qubit.
Most closely comparable to our work is Ref.~\cite{Delfosse_Beverland_Tremblay_2021}, which aimed to implement quantum expander codes with local connectivity by using a similar teleportation-based scheme. Whereas they arrived at a negative result, the innovations in code choice, partial error correction, and syndrome extraction using entanglement purification presented here allow us to obtain more favorable performance.
In general, our approach may be easier to implement as it only requires a bilayer architecture, local connectivity, and relatively few qubits. It, of course, also comes with challenges, which we later discuss.

The paper is structured as follows. In Section~\ref{sec:background}, we give the necessary background on quantum error correction and introduce the architecture and routing assumptions we consider throughout the work. 
We also review the stacked model and motivate the use of masking.
In Section~\ref{sec:bounds}, we develop lower bounds on the routing time for our specific routing model and provide a greedy algorithm to use in implementations.
Section~\ref{sec:circuit_level_sims} develops an error correction protocol built on a bilayer architecture and culminates with circuit-level simulations comparing the performance with the rotated surface code.
We conclude in Section~\ref{sec:discussion} with a discussion. 

\section{Background}
\label{sec:background}

\subsection{Quantum error correction}

\begin{figure}[t]
    \centering
    \includegraphics[width=\linewidth]{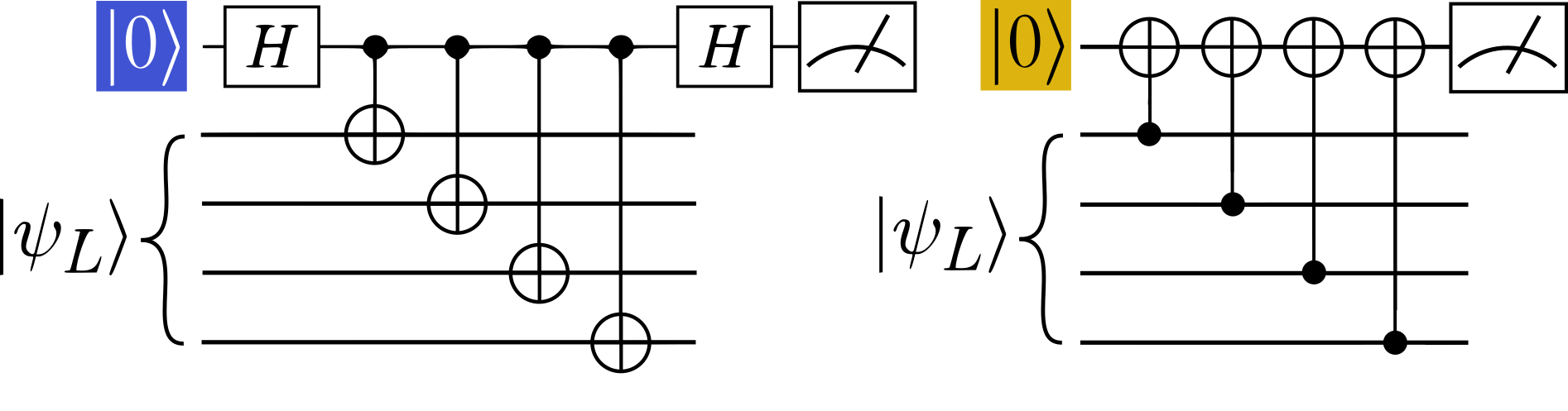}
    \caption{Circuits for measuring the eigenvalue of an $X$-type generator (blue) and a $Z$-type generator (yellow). The $Z$-type measurement presented here is a variation from the standard circuit which uses CZ gates.}
    \label{fig:synd_meas}
\end{figure}

Quantum error correcting codes~\cite{shor1995} 
are believed to be necessary in order to run high-fidelity computations on noisy quantum computers. Without them, errors would accumulate throughout the course of a circuit and render the output unreliable. At a high level, quantum error correcting codes allow us to redundantly encode quantum information in a subspace of the full $2^n$-dimensional Hilbert space and occasionally check to see if errors have caused the information to leave this logical subspace.

\textit{Stabilizer codes} \cite{Gottesman_1997, Calderbank_Rains_Shor_Sloane_1997} are a class of quantum error correcting codes defined by their \textit{stabilizer}, an abelian subgroup of the Pauli group on $n$ qubits that leaves the codespace invariant. Equivalently, the codespace of a stabilizer code is the joint $+1$-eigenspace of the generators of the stabilizer $S = \langle S_1, S_2, \ldots, S_r \rangle$. For a quantum $[[n,k,d]]$ code with $n$ physical qubits, $k$ logical qubits, and distance $d$, the number of linearly independent generators is $r = n-k$. 
A stabilizer code is considered to be a quantum low-density parity-check (qLDPC) code if each qubit is in the support of at most $c$ stabilizer generators and each generator has weight at most $c$, where $c$ is a constant independent of $n$. A stabilizer code is said to be a CSS code~\cite{Calderbank_1996, steane1996} if each generator is a tensor product of $X$ and $I$ or a tensor product of $Z$ and $I$. 
Although the surface code is LDPC, the encoding rate $k/n$ vanishes in the limit as $n \rightarrow \infty$, contributing to its high overhead. Alternative qLDPC codes have asymptotically constant encoding rates while maintaining or improving the $\Theta(\sqrt{n})$ distance scaling of the surface code~\cite{Tillich_2014, hastings2020fiber, Breuckmann_2021_2, Panteleev_2022, panteleev2022asymptotically, leverrier2022quantum, lin2022good}.

From a stabilizer description of a quantum error correcting code, one can define its \textit{Tanner graph} $T(S) = (V_q \ \sqcup \ V_S, E)$. There is a vertex $q \in V_q$ for each data qubit and a vertex $s \in V_S$ for each stabilizer generator.
Two vertices $q \in V_q, s \in V_S$ share an edge $(q, s) \in E$ if the generator $s$ acts non-trivially on qubit $q$.
The Tanner graph of a qLDPC code has degree at most a constant $c$.  

To determine whether the encoded quantum information has left the logical subspace, the eigenvalues of the stabilizer generators are measured.
There are several ways to do this. The circuits depicted in Fig.~\ref{fig:synd_meas} provide one of the most straightforward approaches, which we use throughout the paper. 
As the $n$ data qubits are assumed to be in a codestate, we expect a $+1$ result when the ancillary check qubit is measured. 
A $-1$ result indicates an error that anticommutes with that specific generator.
These measurement results constitute a classical syndrome which is then used as input to a decoding algorithm to correct the errors.

\subsection{Architecture}
\label{sec:arch}
In this work, we consider an architecture where qubits are located on the vertices of an $M \times M$ grid, where $M = \Theta(\sqrt{n})$. As is natural for current superconducting quantum computing platforms, we assume that two-qubit gates can only be performed between neighboring qubits on the grid. 
Any two-qubit gate which interacts qubits that are not neighboring is considered a \textit{long-range} gate.
Circuits that do not have access to long-range gates are called \textit{2D local circuits}, and architectures that are restricted to these circuits are called \textit{2D local architectures}.
This definition generalizes to architectures based on graphs other than the grid: given a connectivity graph $G = (V,E)$ with data qubits located on the vertices, the only allowed two-qubit gates are those between qubits $u,v \in V$ that share an edge $(u,v) \in E$.
Similar restrictions arise if we disallow the slow movement of atoms in neutral-atom devices, in which case the only available two-qubit gates are those performed through Rydberg-Rydberg interactions. This leads to an architecture that can perform entangling gates on qubits that are some distance $R$ away, where $R$ depends on the capabilities of the device. We do not investigate this ability in this work, but we discuss it in Section~\ref{sec:discussion}.

\begin{figure*}
    \centering
    \includegraphics[width=\textwidth]{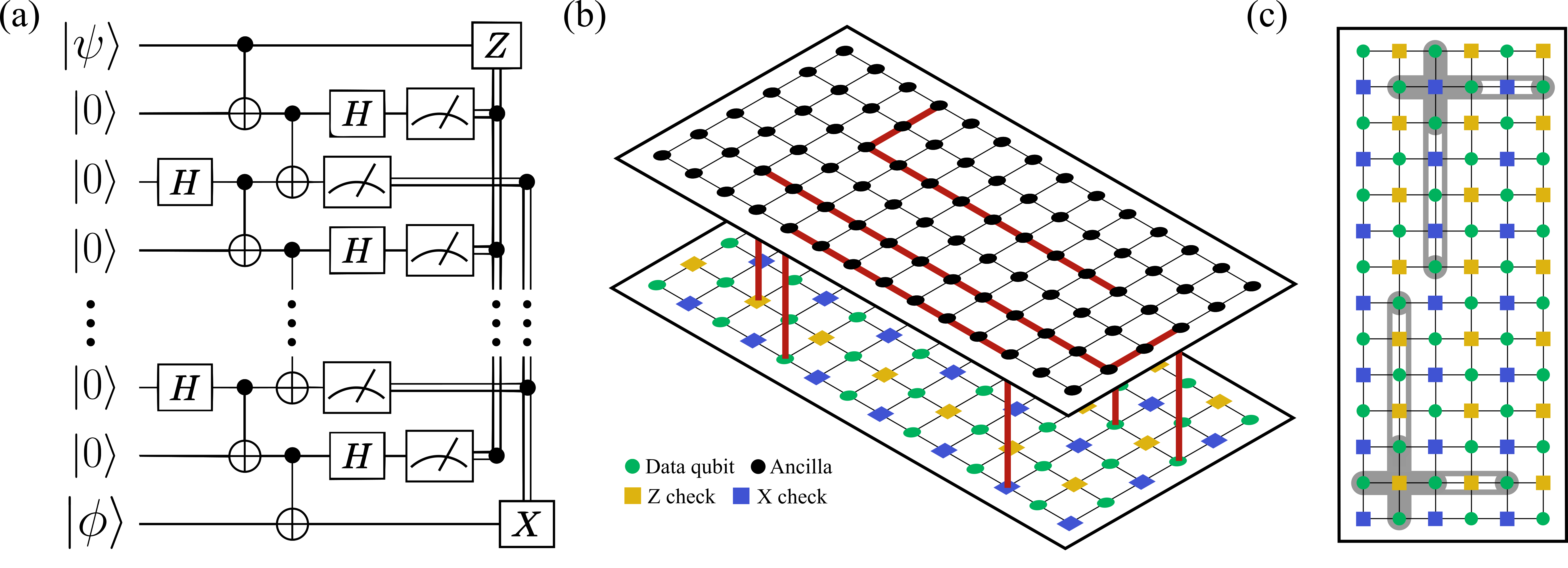}
    \caption{(a) Circuit to teleport a CNOT gate through a chain of $n$ qubits using only 2D local gates. The depth of the circuit is constant regardless of the length of the chain. 
    (b) Proposed architecture to implement nonlocal high-rate qLDPC codes using only 2D local gates. The architecture consists of two qubit layers: the bottom layer contains the data qubits and ancilla qubits allocated to perform syndrome measurements, while the top layer contains extra ancilla qubits used to perform long-range CNOT gates. Each layer has only 2D local connections, and the only connections between the layers are between qubits that are vertically adjacent. To perform a CNOT gate on two spatially distant qubits, the circuit from (a) is used along the paths of qubits highlighted in red. Multiple long-range CNOT gates may be performed in parallel, as long as the paths act on disjoint sets of qubits. 
    (c) Example embedding for a $[[42,12,2]]$ BB (error detecting) code constructed with $\ell=7, m=3$ and by matrices $A = 1 + y^2 + y$, $B = 1 + x^5 + x$. The check structure, which is identical for all checks of both types up to mirroring, translation, and boundary conditions, is highlighted in gray. }
    \label{fig:architecture}
\end{figure*}

Implementing general quantum circuits on real architectures requires compilation into a form that respects the connectivity constraints of the device.
For the 2D local architecture we consider here, performing two-qubit operations on qubits that are not adjacent requires permuting them to be so. Doing this with swap gates requires a circuit depth proportional to the distance between the qubits.
To implement stabilizer generator measurements like those shown in Fig.~\ref{fig:synd_meas}, this means that each data qubit must be moved to a position where it can interact with the check qubit, so one must wait for these permutations to complete before the eigenvalue can be measured.
This somewhat defeats the purpose of using qLDPC codes, since a single syndrome can no longer be extracted with a constant-depth circuit. 
As such, it is infeasible to perform long-range stabilizer generator measurements in this way, and we instead focus on an alternative method.

\subsection{Teleportation routing}
\label{sec:routing}
\textit{Routing} is the task of permuting packets of information, or tokens, on the vertices of a graph, using only interactions on edges of the graph. In quantum routing, the tokens are qubits, and the graph is specified by the architecture's connectivity constraints. Classical approaches to routing are typically built from swap gates \cite{Alon1994, Zhang, Chung1996}, which can also be applied naturally to routing quantum data \cite{Cowtan_2019, Childs_Schoute_Unsal_2019}. However, more general quantum operations can enable faster routing. In particular, measurement and classical feedback enable the use of entanglement swapping to distribute entanglement and perform quantum teleportation, which can achieve speed-ups over swap-based routing for many permutations and underlying graphs \cite{devulapalli2022quantum, rosenbaum, hillmich2021}, with applications including error correction~\cite{Beverland_Kliuchnikov_Schoute_2022}.

We assume the \emph{LOCC routing} model described by \textcite{devulapalli2022quantum},
where arbitrary single-qubit and disjoint two-qubit quantum gates can be implemented in a single time step,
and we have access to fast single-qubit, mid-circuit measurements, and fast classical control of single-qubit gates.
Additionally, there are a constant number of ancillary qubits for each data and check qubit, connected as attached ancillas~\cite{loss1998, lukin2007} or through stacked vertical layers (see Section~\ref{sec:synd_ext_cir}).
In LOCC routing, we can perform protocols such as entanglement swapping~\cite{PhysRevLett.71.4287} and teleportation in constant depth.
A specialization of LOCC routing that focuses on qubit and gate teleportation~\cite{Gottesman_1999} is teleportation routing.
During a single round of teleportation, we perform parallel entanglement swapping along multiple teleportation paths.
Each vertex can be involved in at most a constant number of paths, as we allow a constant number of ancillary qubits per vertex. 
In this work, we assume only one ancilla per data qubit and use the stacked vertical layers model.
This model allows direct implementation of gates between ancillas and their corresponding data qubits, as well as between ancillas whose data qubits are also directly connected (see Fig.~\ref{fig:architecture}(b)).

To perform long-range two-qubit gates, it is not necessary to actually teleport the participating qubits to adjacent locations; instead, it suffices to use the teleportation paths to implement a long-range gate with gate teleportation. The circuit shown in Fig.~\ref{fig:architecture}(a) allows us to implement arbitrarily long CNOT gates in constant quantum depth, avoiding depth overhead of swap routing and any need to reverse the operation. At the cost of utilizing ancillary qubits, this lets us extract the syndrome of a single nonlocal generator using only a constant-depth circuit.

\subsection{Stacked model}
\label{sec:stacked}
The stacked model has recently been introduced as a potential avenue to reduce overhead when implementing qLDPC codes in architectures restricted to 2D local gates~\cite{Baspin_2022, berthusen2023}.
In the stacked model, the stabilizer generators of a quantum error correcting code are partitioned into several layers depending on the size of the ball containing the qubits in its support. The lowest layer of the stack contains generators that are local, and the higher layers contain nonlocal generators whose interaction radius is some function of the system size. 
For certain codes, most of the generators are located at the bottom of the stack, i.e.~mostly local, whereas only a small fraction are large.
When restricted to 2D local gates, the set of nonlocal generators takes much longer to route and measure than the local generators. Measuring the nonlocal generators less frequently than the local ones could significantly shorten the syndrome extraction time, at the cost of potentially reduced error correction capabilities. 
Note that the layers in the stack do not correspond to physical layers on hardware; instead, they are a conceptual tool for partitioning the generators into sets based on their geometric size.

The concept of \textit{masking}~\cite{gottesman2022opportunities, berthusen2023} formalizes using an incomplete set of generators to perform error correction. 
Measuring a subset of stabilizer generators corresponds to choosing a subgroup of the stabilizer $T \subseteq S$, and the stabilizer generators that are not measured, $S \char`\\ T$, are considered to be masked. 
Error correction performance may be degraded since the resulting syndrome may have less information about the error than would be available by measuring the full set of generators.
During a circuit with $t = 1,\dots,\tau$ error correction rounds, we specify a subgroup $T_t \subseteq S$ for each round; or equivalently, we specify the generators of $S \char`\\ T_t$ that are masked. 
For generators that were previously masked, unmasking them adds them into the current subgroup, and their eigenvalues are able to be measured. A single generator may be masked and subsequently unmasked many times over the course of a circuit.

An important consideration for this model is the specific assignment of physical qubits in the architecture to data and check qubits in the code, which can be considered a type of qubit placement~\cite{shafaei2014qubit} or qubit allocation~\cite{siraichi2018qubit}.
This assignment can be thought of as an \textit{embedding} of the Tanner graph of the code in the architecture, where an embedding for a graph $G=(V,E)$ is a map $\eta: V \rightarrow \Z^D$. 
As an example, the Tanner graph for the surface code has a natural embedding into $\Z^2$ that allows for all of its generators to act on qubits within a constant radius; however, one could instead assign data and check qubits to physical qubits randomly, yielding generators that still have weight four, but are no longer local. The difficulty of implementing syndrome extraction circuits is closely related to the chosen embedding.
In Section~\ref{sec:codes}, we discuss the embedding problem for a specific class of codes.

To study the impact of nonlocality on the cost of performing syndrome measurement, we must quantify the notion of generator size and size frequency. 
We parameterize the size of a given generator as $M^\gamma$, where $0 \le \gamma \le 1$ and $M$ is the linear size of the grid. For local generators, $M^\gamma = O(1)$ implies a constant interaction radius, while the largest generators can have interaction radii $\frac{\sqrt{2}}{2}M \in \Theta(M)$ (i.e., $\gamma = 1$). 
For stabilizer codes, the number of independent stabilizer generators $r$ is related to the number of physical and logical qubits in the code like $n - r = k$. For constant-rate codes, there are thus $O(n) = O(M^2)$ independent generators, which can be parameterized like $M^{2\beta}$, with $0 \le \beta \le 1$. With $\beta=1$, we are considering the problem of measuring every generator, and with $\beta < 1$ we only consider some subset.
We can describe the set of generators as a whole by defining a function $f(\gamma)$ to characterize the distribution of generators having size $M^\gamma$. The only constraint on $f(\gamma)$ is that it is a valid probability distribution over the domain of $\gamma$, $\int_0^1 f(\gamma) d\gamma = 1$. In practice, $f(\gamma)$ will depend on the architecture, embedding, and parameters of the code family of interest~\cite{Baspin_2022, Baspin_2022_2}.

A rough estimate of the amount of work required to perform the syndrome extraction circuits for a given set of generators is simply to count the two-qubit gates, which in many cases is the leading contributor to the error budget.
In our routing model, this value is proportional to the total length of the teleportation paths when implementing long-range CNOT gates, which can be approximated as
\begin{equation}
    \text{total path length} \approx M^2\int_0^1 f(\gamma) M^\gamma d\gamma.
\end{equation}
Here, the $M^2$ factor comes from the fact that there are $O(M^2)$ generators to measure in total, and a single generator of size $\gamma$ requires a path length of $M^\gamma$. If we choose to only measure generators below a certain size $\gamma'$, this corresponds to simply evaluating the integral up to $\gamma'$. We might also want to consider measuring the smallest $x \%$ of generators, in which case one can solve $x = 100\int_0^{\gamma'} f(\gamma) d\gamma$ to find the appropriate value of $\gamma'$ and then proceed in the same way.

\section{Routing bounds \label{sec:bounds}}

Previous work by \textcite{Delfosse_Beverland_Tremblay_2021} developed lower bounds on the depth of Clifford circuits required to measure commuting Pauli operators. 
In this section, we derive similar bounds taking advantage of additional information about the geometric size of the operators. 
These bounds do not hold in general, but are instead specific to the teleportation routing model discussed in Section~\ref{sec:routing}. 
We assume there is a fixed layout of the data and check qubits that gives rise to a specific generator size distribution $f(\gamma)$. This is to avoid scenarios such as scrambled surface codes, where the difficulty of implementing the syndrome extraction circuits could be greatly reduced by permuting the qubits.

\begin{claim}
\label{claim:basic_routing}
    Let $C$ be a 2D local circuit measuring $M^{2\beta}$ commuting Pauli operators whose radii are greater than $M^\gamma$ after embedding them in an $M \times M$ grid. Then for teleportation routing,
    \begin{equation}
        \depth(C) = \Omega\big(M^{2\beta + \gamma - 2}\big).
    \end{equation}
\end{claim}

\begin{proof}
In our routing model, the maximum total length of the teleportation paths in a single time step is $O(M^2)$ since only a constant number of ancillary qubits per data qubit are allowed, and there are $\Theta(M^2)$ edges in the grid graph.
The cost of measuring an operator of size $\Omega(M^\gamma)$ is dominated by implementing the long-range CNOT gate between its two furthest qubits. Although this can be done in constant depth using a dynamic circuit (Fig.~\ref{fig:architecture}(a)), it requires a teleportation path of length $\Omega(M^\gamma)$. Consequently, routing and measuring this one operator uses $\Omega(M^\gamma)$ edges of the $O(M^2)$ available edges. 
Measuring all $M^{2\beta}$ operators thus requires $\Omega(M^{2\beta + \gamma})$ edges. 
In the best case, we utilize all available edges in each circuit layer, giving a circuit depth of $\Omega(M^{2\beta + \gamma - 2})$.
\end{proof}



In practice, it will often be the case that the edges are not optimally used, as illustrated in Fig.~\ref{fig:routing-algo-results}.
We can extend this idea to the general case of an arbitrary distribution of generator sizes.

\begin{claim}
    \label{claim:dist_routing}
    Let $C$ be a 2D local circuit measuring $M^{2\beta}$ commuting Pauli operators whose radii follow a probability distribution $f(\gamma)$ after embedding them in an $M \times M$ grid. Then for teleportation routing,
    \begin{equation}
        \label{eq:dist_eq}
        \depth(C) = \Omega\big(M^{2\beta - 2} \int_0^1 f(\gamma) M^\gamma d\gamma\big).
    \end{equation}
\end{claim}

\begin{proof}
Just as in Claim~\ref{claim:basic_routing}, we can lower bound the circuit depth by summing the lengths of the teleportation paths required to measure the set of operators. We now have operators of different sizes, where the fraction of operators of a certain size is determined by the probability distribution $f(\gamma)$. 

Thus, for a given $\gamma$, there are a number of operators proportional to $f(\gamma)M^{2\beta}$ that each require $M^\gamma$ edges to measure. Since $0 \le \gamma \le 1$, the total teleportation path length needed to route and measure every operator is 
\begin{equation}
M^{2\beta} \int_0^1 f(\gamma) M^\gamma d\gamma.    
\end{equation}
Since we again have $O(M^2)$ edges in the grid available in each layer of the circuit, the total circuit depth is lower-bounded as in Eq.~\eqref{eq:dist_eq}, as desired.
\end{proof}


\subsection{Greedy routing}
\label{sec:greedy}
Swap routing is a straightforward approach to compiling circuits for quantum hardware with interaction constraints.
Practically, this can be done using an algorithm that tries to perform the circuit using as few swap gates as possible~\cite{li2019tackling, Childs2019, Cowtan_2019, herbert2020depth}.  
As mid-circuit measurement and long-range entanglement generation become more reliable~\cite{baumer2023efficient}, teleportation routing may become a more viable option to move qubits and perform long-range gates. 
Here, we present a simple, greedy algorithm to route an arbitrary set of operators under the routing and architecture assumptions of Sections~\ref{sec:routing} and~\ref{sec:arch}, respectively.
An operator consisting of a tensor product of single-qubit Paulis, such as a stabilizer generator, can only be measured once each qubit in its support has been routed. That is, a teleportation path is prepared and a long-range entangling gate is applied between the qubit and a readout ancilla qubit. 
Once all required gates have been applied, the operator is said to have completed routing, and the readout qubit can be measured to obtain the eigenvalue of the operator.
The algorithm is described below in Algorithm~\ref{alg:greedy}. 

\begin{algorithm}[H] 
	\caption{Greedy routing} 
    \label{alg:greedy}
	\begin{algorithmic}[1]
        \While {there are still operators to measure}
            \State Sort the operators in decreasing order according to how many of their qubits have completed routing.
            \For{incomplete operator $o_i = 1,2,\ldots$}
                \For{qubits $j = 1,2,\ldots$ of operator $o_i$}
                    \State Use breadth-first search to find a teleportation path for qubit $j$ to the corresponding readout ancilla qubit.
                    \State If no path exists, continue.
                \EndFor
            \EndFor
            \State Perform long-range entangling gates on qubits that found a teleportation path.
            \State Measure the readout qubit of operators that have completed routing.
        \EndWhile

	\end{algorithmic} 
\end{algorithm}

The circuit operations of a single iteration can be executed in parallel, so each iteration performs only a constant-depth circuit. Therefore, the total circuit depth of the routing procedure is proportional to the number of iterations.
Instead of minimizing gate count, the intent of this algorithm is to minimize the circuit depth---and saturate the bound of Claim~\ref{claim:dist_routing}---by maximizing the usage of teleportation paths. 
This is only possible if the partial measurements between iterations commute, such as when measuring the generators of a single type in a CSS code, in the standard surface code syndrome extraction circuit~\cite{Tomita_2014}, or in the depth-7 BB code measurement circuit~\cite{bravyi2023highthreshold}. 
The syndrome extraction circuits we use route every $Z$-type check and then route every $X$-type check.

\begin{figure}[b]
    \centering
    \includegraphics[width=\linewidth]{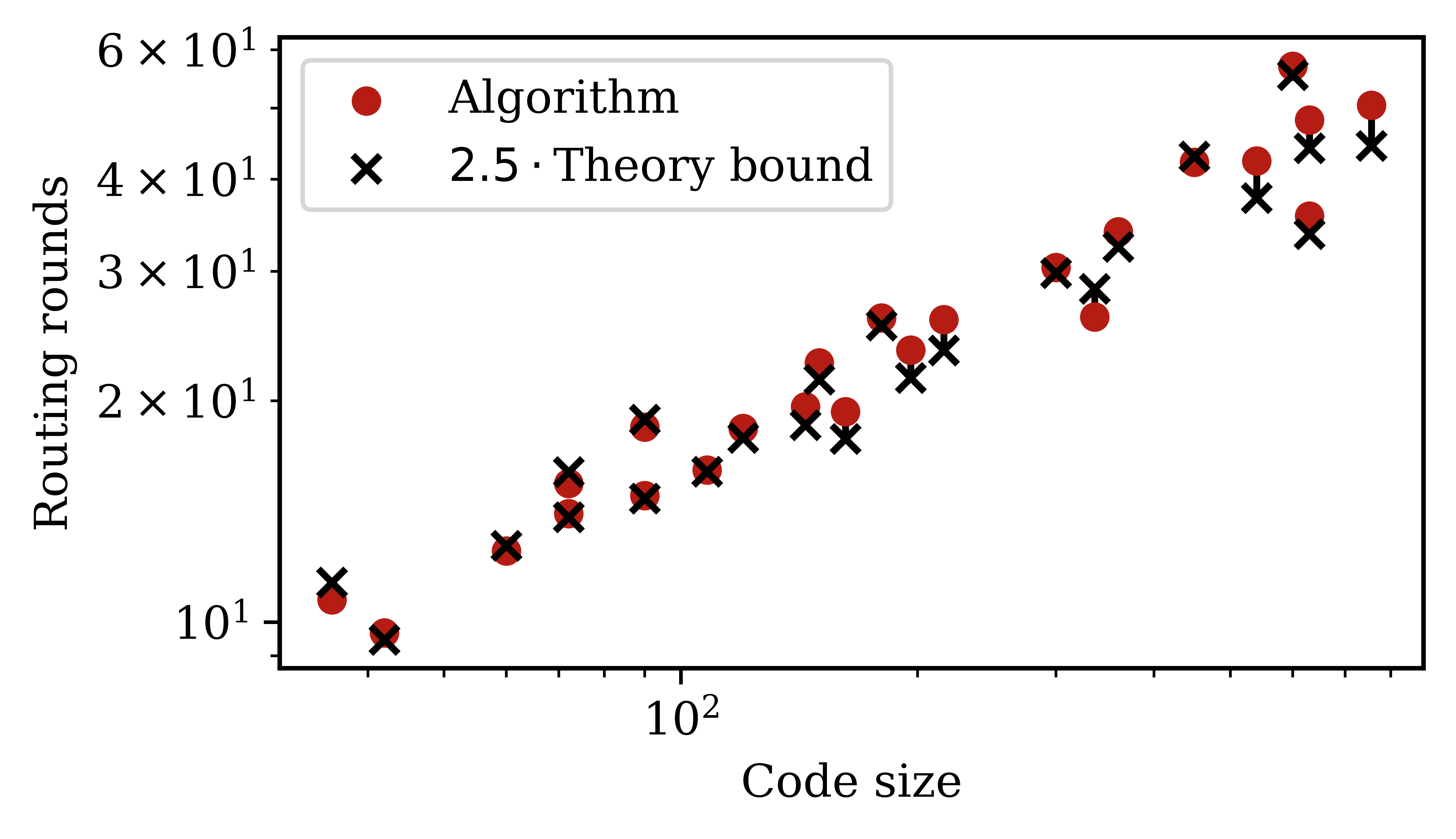}
    \caption{Depth from greedy routing versus $2.5 \times$ the theoretical optimal depth to route the $X$-type generators using a single layer of ancillary qubits. Code examples are drawn randomly from the family of BB qLDPC codes (see Section~\ref{sec:codes}). }
    \label{fig:routing-algo-results}
\end{figure}

To benchmark the performance of the algorithm, we draw random examples of BB codes (see Section~\ref{sec:codes}) and route the $X$-type generators while restricted to a single layer of ancillary qubits.
For comparison, we compute the optimal routing depth according to Claim~\ref{claim:dist_routing}.
Figure~\ref{fig:routing-algo-results} shows the results of these simulations, providing evidence that the greedy routing algorithm nearly saturates Eq.~\eqref{eq:dist_eq}. To obtain the constant multiple in Fig.~\ref{fig:routing-algo-results}, we consider the smallest eight code instances and perform a fit between the asymptotic lower bound and the depth returned from the greedy routing algorithm. This constant times the theory lower bound matches closely with the routing time of small code instances; however, we begin to see the algorithm routing depth deviating as we increase the block length, indicating non-optimal performance. For code sizes of practical interest, this algorithm may be a viable option to optimize teleportation routing. 
Certain codes, such as the BB codes we discuss in the next section, have additional structure that allows us to manually find routing schedules that outperform those found by the greedy algorithm.

\section{Bilayer architecture}
\label{sec:circuit_level_sims}

\subsection{Bivariate bicycle codes}
\label{sec:codes}

\renewcommand{\arraystretch}{1.25}
\setlength{\tabcolsep}{8pt}
\begin{table*}[t]
\begin{tabular}{c|c|c|c|c|c|c}
\hline
$[[n,k,d]]$  & $\ell,m$ & $A$            & $B$   & Embedding & Mask percent    & Routing steps       \\
\hline\hline
$[[72,8,6]]$ & 12,3   & $x^{9}+y^1+y^2$ & $1+x^{1}+x^{11}$ &$\langle A_2A_3^T, B_1B_2^T\rangle$& 25\% & 11,6\\
\hline
$[[90,8,6]]$   & 9,5   & $x^{8}+y^4+y$ & $y^5+x^{8}+x^7$ &$\langle A_2A_3^T, B_2B_1^T\rangle$ & 22.22\% & 9,5\\
\hline
$[[120,8,8]]$ & 12,5 & $x^{10}+y^4+y$ & $1+x+x^2$ &$\langle A_2A_3^T, B_1B_2^T\rangle$& 25\%  & 11,6\\
\hline
$[[150,8,8]]$ & 15,5 & $x^5+y^2+y^3$ & $y^2+x^7+x^6$ & $\langle  A_1A_2^T, B_1B_3^T\rangle$ & 26.66\%  & 11,6  \\
\hline
$[[144,12,12]]$ & 12,6 & $x^3+y+y^2$ & $y^3+x+x^2$ & $\langle  A_2A_1^T, B_1B_3^T\rangle$ & 33.33\%  & 12,8 \\
\hline
$[[196,12,8]]$ & 14,7 & $x^6+y^5+y^6$ & $1+x^4+x^{13}$ & $\langle  A_2A_3^T, B_1B_2^T\rangle$ & 35.71\%  &  16,15\\
\hline
\end{tabular}
\caption{Examples of BB qLDPC codes found through a computer search. Code distances were computed using the QDistRnd GAP package~\cite{Pryadko_2022}, with 1000 information sets and \textsc{mindist = 0} to obtain the actual distance. The `Embedding' column reports the specific embedding into $\mathbb{Z}^2$ used for that code (see Appendix~\ref{apx:embedding}). The `Mask percent' column denotes the percentage of generators that are ``large", i.e., induced by the long boundary and masked during a portion of the error correction rounds. The `Routing' steps column indicates the number of routing rounds required to route, purify, and measure the short-range and long-range generators, respectively. Algorithm~\ref{alg:greedy} was not used to determine the circuits; instead, the repeated generator structure of the BB codes allowed us to find circuits by hand. The actual circuit depth is $11\times$ greater due to the Bell pair generation (depth 6), purification (depth 2), and implementation of the long-range CNOT gate (depth 3).}
\label{tab:codes}
\end{table*}

In this work, we investigate the recently introduced bivariate bicycle qLDPC codes~\cite{bravyi2023highthreshold}, which come from the wider family of generalized bicycle codes~\cite{Kovalev2013}.
Let $\mathbb{I}_\ell$ be the $\ell \times \ell$ identity matrix and let $S_\ell$ be the $\ell \times \ell$ cyclic permutation matrix, which is obtained by shifting the columns of $\mathbb{I}_\ell$ one position to the right. Also let
\begin{equation}
    \label{eq:xy}
    x = S_\ell \otimes \mathbb{I}_m \quad\text{and}\quad y = \mathbb{I}_\ell \otimes S_m
\end{equation}
for integers $\ell, m$. We then define two matrices
\begin{equation}
    \label{eq:ab}
    A = A_1 + A_2 + A_3 \quad\text{and}\quad B = B_1 + B_2 + B_3
\end{equation}
where $A_i, B_i$ are powers of $x$ or $y$. Here we perform all arithmetic over $\mathbb{Z}_2$. Using $A$ and $B$, we can construct the CSS-type BB code $QC(A,B)$ with $X$- and $Z$-parity checks that, respectively, take the form
\begin{equation}
    \label{eq:parity_check}
    H_X = [A | B] \quad\text{and}\quad H_Z = [B^T | A^T].
\end{equation}
To define a valid stabilizer code, we require that all $X$-type checks commute with all $Z$-type checks, which translates to the condition $H_X \cdot H_Z^T = AB + BA = 0$. Since $[x,y] = 0$, this condition is satisfied. 

For certain choices of $A_i$ and $B_i$, the resulting BB code has an embedding into $\mathbb{Z}^2$ that yields checks which act on four nearest-neighbor qubits and two distant qubits (see Appendix~\ref{apx:embedding}). 
Another useful property of generalized bicycle codes is the repeated parity check structure: given one check, other checks of the same type can be obtained with vertical and horizontal shifts on the grid, up to periodic boundary conditions. Opposite-type checks are obtained by mirroring and again performing horizontal and vertical shifts. 
Fig.~\ref{fig:architecture}(c) shows an example of an embedding for a $[[42,12,2]]$ code constructed with $\ell=7, m=3$ and by matrices $A = 1 + y^2 + y$, $B = 1 + x^5 + x$. The check structure for the weight-6 $X$- and $Z$-type generators is indicated by the gray outline.

These natural embeddings make it straightforward to search for codes where the check structure is geometrically small.
While the checks are not entirely local due to the two nonlocal qubits in their support, appropriately choosing $\ell$ and $m$ can make the periodic boundary conditions induce generators that are comparatively much larger. This can be done by letting $\ell \gg m$, as illustrated in Fig.~\ref{fig:lrsr}. 
In the resulting generator distribution, the majority of the checks are geometrically small.
In the context of the stacked model, the generators that are induced by the boundary conditions are those that are measured less frequently. 

Table~\ref{tab:codes} lists BB codes found by computer search which, through simulations similar to those of Ref.~\cite{berthusen2023}, display good numerical performance. For each code, every valid embedding was simulated in a simplified version of the protocol in Section~\ref{sec:results} in order to find the embedding that yielded the best masked error correction performance. Choosing an embedding determined the percentage of generators induced by the long boundary. This percentage is listed in Table~\ref{tab:codes} in the Mask percent column.
To our knowledge, the codes presented here are new, with the exception of the $[[144,12,12]]$ code, which was reported in Ref.~\cite{bravyi2023highthreshold}.

\subsection{Syndrome extraction circuits}
\label{sec:synd_ext_cir}

As detailed in Section~\ref{sec:arch}, the main difficulty in implementing nonlocal qLDPC codes on 2D local architectures is the need to perform nonlocal two-qubit operations. To address this issue, we propose a physical implementation based on the teleportation routing model described in Section~\ref{sec:routing}. The architecture, as depicted in Fig.~\ref{fig:architecture}(b), consists of two layers of qubits. The bottom layer contains the data qubits and ancillary qubits to perform syndrome measurements (check qubits), laid out using an embedding that maximizes decoding performance while minimizing the number of long-range generators. The top layer contains ancilla qubits to aid in the implementation of long-range CNOT gates. In each layer, the only allowed two-qubit operations are between neighboring qubits, and operations between layers are only allowed between qubits that are vertically adjacent, i.e., at the same $(x,y)$ location. 

A bilayer architecture is a feasible design requirement for several types of quantum computers.
As discussed in Ref.~\cite{bravyi2023highthreshold}, it is difficult, yet not unreasonably so, to modify the current generation of superconducting hardware to support a second layer. 
In movement-restricted neutral-atom devices, one option is to use dual-species Rydberg arrays~\cite{Zeng_2017, Singh_2022, anand2024dualspecies}, where the data layer is made up of one species, and the ancilla layer the other.
Alternatively, for single-species arrays, it may be practical to store multiple qubits per atom, using a combination of nuclear and electronic~\cite{Gorshkov_2009, lis2023midcircuit} or motional qubits~\cite{scholl2023erasurecooling}.

To implement a CNOT gate between a data qubit and a distant check qubit, we use the constant-depth circuit shown in Fig.~\ref{fig:architecture}(a). A number of ancilla qubits equal to the length of the CNOT gate are needed, and so qubits from the upper layer are utilized, as illustrated in Fig.~\ref{fig:architecture}(b). Multiple long-range CNOT gates may be performed in parallel as long as the paths act on disjoint sets of qubits. Given a set of CNOT gates to perform, an order that attempts to minimize the total depth of the circuits can be found using the greedy routing algorithm introduced in Section~\ref{sec:greedy}. 
Alternatively, we can utilize the repeated check structure of the BB codes to manually come up with highly parallelized orderings;
Fig.~\ref{fig:purifcation} shows an example of how the Bell pairs needed for the long-range CNOT gates (red highlighted paths) can be implemented in parallel. Also see Fig.~\ref{fig:routing_schedule}. The `Routing steps' column in Table~\ref{tab:codes} indicates the number of routing rounds required to route, purify, and measure the short-range and long-range generators, respectively, for these hand-designed orderings.
This repeated parity check structure is also useful for implementing generalized bicycle codes with reconfigurable atom arrays~\cite{viszlai2023matching} and bosonic cat qubits~\cite{ruiz2024ldpccat}.

Remote CNOT gates implemented in this way have an error rate proportional to the length of the chain. For short distances, the resulting error rate is not much worse than the native two-qubit CNOT error rate; however, larger chains will be prohibitively noisy. To remedy this, we can apply entanglement purification~\cite{Bennett_1996, Deutsch_1996} to the noisy Bell pairs in the ancilla layer. Figure~\ref{fig:purifcation} outlines the original purification scheme as proposed by Bennett \textit{et al.}~\cite{Bennett_1996}. The protocol uses additional `donor' Bell pairs (pink highlighted paths) to create `source' Bell pairs (red highlighted paths) with higher fidelity. This is done by performing CNOT gates between the ends of the source and donor pairs, measuring the ends of the donor pairs in the computational basis, and then comparing the measurement results classically. If the results agree, the source Bell pair is kept; otherwise it is discarded. Averaging over cases where the source Bell pair is kept, it has a higher fidelity than an unpurified pair; however, in the cases where it is discarded, the corresponding long-range CNOT gate cannot be performed. 
We either have the option of reattempting the purification process, implementing the long-range CNOT with the flawed Bell pair, or giving up on the long-range CNOT gate (and ultimately the corresponding generator syndrome measurement) altogether. Since we already intend not to measure every generator at every error correction round, this last option is most appropriate. In the context of the bilayer architecture, both donor and source Bell pairs are routed through the ancilla layer. Practically, this means that fewer long-range CNOT gates can be implemented in parallel, since the purification process uses additional teleportation paths.

\begin{figure}[t]
    \centering
    \includegraphics[width=\linewidth]{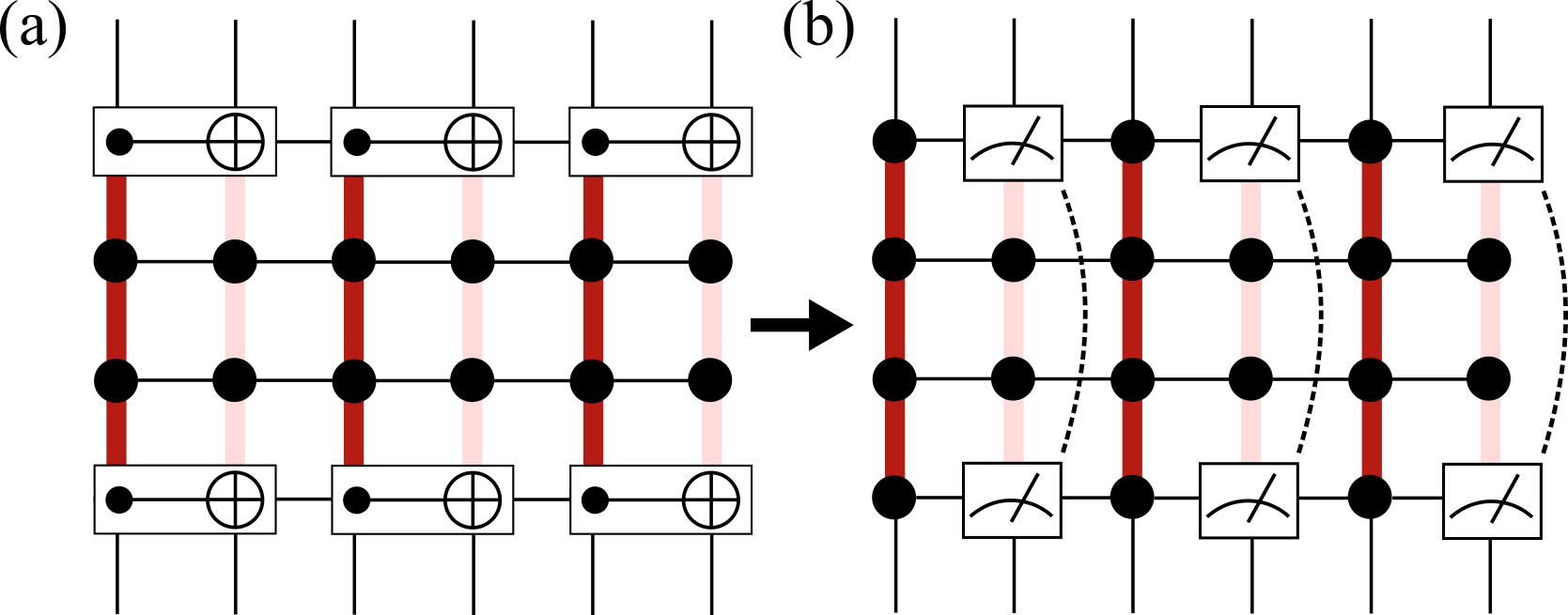}
    \caption{Implementing multiple long-range Bell pairs in parallel for a BB code. The `source' red highlighted Bell pairs are purified using the Bennett protocol~\cite{Bennett_1996}. (a) CNOT gates are performed between each end of the source and the pink `donor' Bell pairs. (b) Each end of the donor Bell pair is measured and the results compared classically. If the measurements agree, the source Bell pair is kept and used; otherwise it is discarded. }
    \label{fig:purifcation}
\end{figure}

Although we now have a way to implement long-range CNOT gates, measuring every stabilizer generator in this manner incurs additional overhead (see Section~\ref{sec:results}). Instead, we can reduce the time overhead by applying the stacked model and choosing to measure the costly, large generators less frequently than the smaller ones. The frequency at which the long-range generators are measured can be tuned, with more frequent measurements potentially correcting more errors but increasing the time needed to implement error correction.

In the phenomenological noise model, depolarizing errors are introduced with probability $p$ only at the beginning of each error correction round. The syndrome is then noiselessly computed using the parity check matrix and the randomly drawn errors. Additional errors may be introduced to the syndrome to represent measurement errors. To correct the qubit errors, the syndrome is given as input to a decoder which attempts to deduce the most likely error. Decoding is considered a success if the guessed error is equivalent to the actual error up to a stabilizer element.

In reality, errors may occur at any operation in the syndrome extraction circuit, including qubit initialization, single- and two-qubit gates, measurements, and idle locations. To model this, we instead consider the standard circuit-based depolarizing noise model~\cite{Fowler2009}, where for each operation in the circuit, an error is introduced with some probability $p$. For example, an error arising from a CNOT gate is the gate followed by one of the possible 15 non-identity two-qubit Pauli products on the control and target qubits. Although it is possible to decode circuit-level noise using the same method as for phenomological noise, it has been shown to be advantageous to instead use a space-time circuit-level decoder~\cite{xu2023constantoverhead, Wang_2011}. Here, the goal is to guess the error at specific locations in the syndrome extraction circuit. Decoding is considered a success if the guessed errors have the same effect on the logical observables as the actual error.

The input to the space-time decoder is not the syndrome of the error, but rather the parities of the syndrome measurements between error correction rounds.
In the absence of errors, the syndrome between rounds should be constant, i.e.~have parity of zero. A parity of one indicates that an error occurred at some point in the previous error correction round. Following the notation of Stim~\cite{gidney2021stim, McEwen_2023}, we define the $i$th \textit{detector} at time $t$ to be the parity of the syndrome of the current and previous rounds $D_i^{(t)} = \sigma_i^{(t)} \oplus \sigma_i^{(t-1)}$, where $\sigma_i^{(t)}$ is the $i$th bit of the syndrome at time $t$. However, in the stacked model, we have the possibility of neglecting to measure certain generators for some number of rounds, $t_m$. As such, detectors for these generators must compare the parities of the corresponding syndromes $t_m$ rounds apart, $D_i^{(t)} = \sigma_i^{(t)} \oplus \sigma_i^{(t-t_m)}$. Each detector allows us to determine whether errors have occurred in a specific detecting region~\cite{McEwen_2023} of the circuit. Figure~\ref{fig:detecting}(a) shows a simple example of a classical repetition code circuit with its associated detectors and highlighted detecting regions. 

\subsection{Space-time decoder}
\label{sec:space_time}
\begin{figure}[!t]
    \centering
    \includegraphics[width=\linewidth]{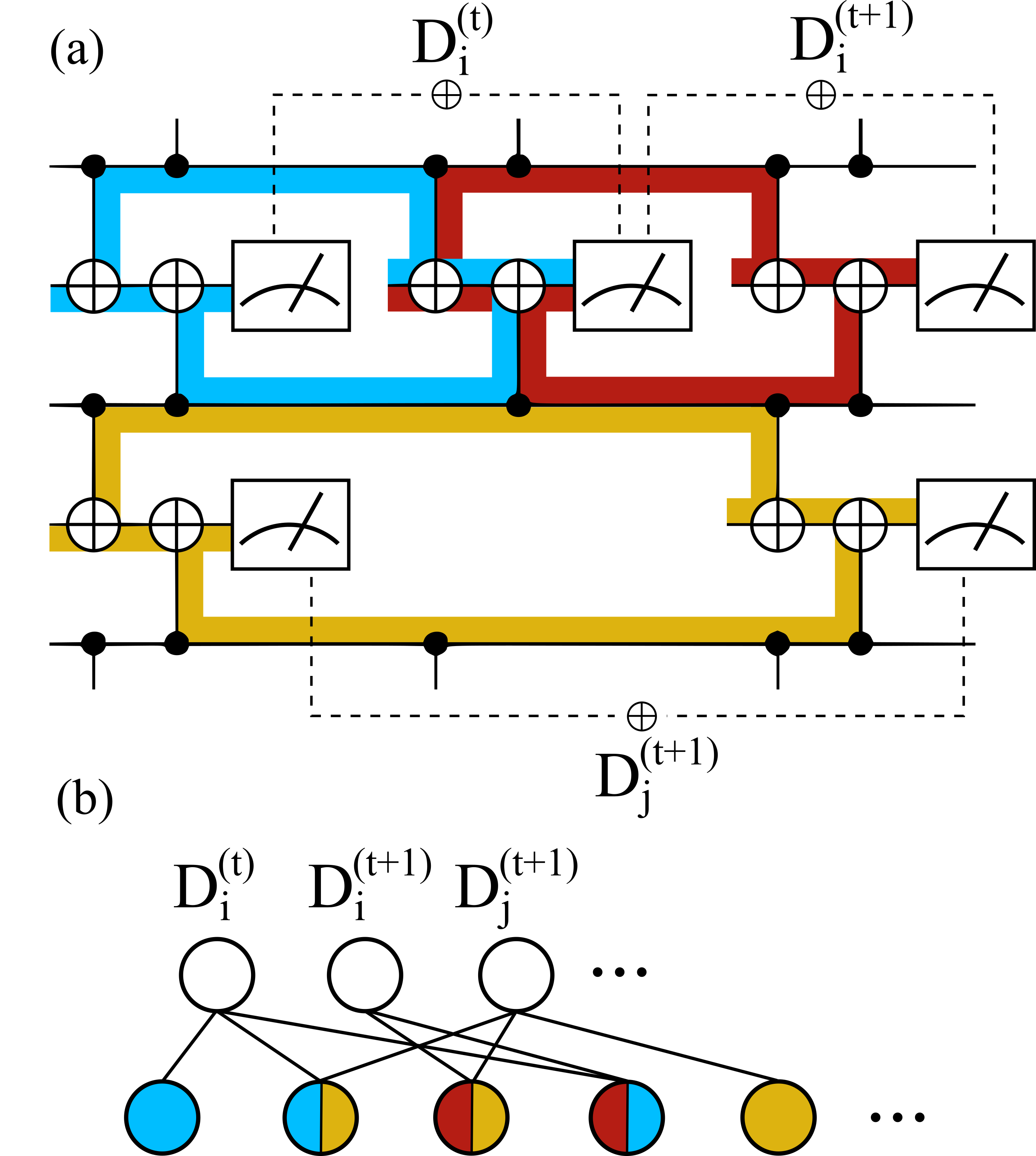}
    \caption{(a) Detectors for a portion of a bit-flip repetition code. The highlighted regions represent the detecting region~\cite{McEwen_2023} of a detector, the set of errors that would cause the detector to be triggered. The corresponding detectors are then the parities of the measurements in that region. Since syndrome $j$ was masked for a round, the detector now represents the parities of the measurements in the region that spans three rounds. (b) The bipartite space-time decoding graph of the circuit. The check nodes of this graph are the detectors, and the bit nodes are possible errors during the execution of the circuit. A detector and error are connected by an edge if the error causes the detector to be activated. Errors on the boundary of two detecting regions cause both detectors to trigger.}
    \label{fig:detecting}
\end{figure}

To correct for errors in the circuit-level model, we relate the detectors with errors in the circuit by constructing a bipartite graph.
Let the detectors over $T$ rounds be the check nodes, and let every possible single- and two-qubit error over the circuit make up the bit nodes. A detector and error are connected by an edge if the error causes the detector to activate. As a practical note, many errors have the same action on the detectors and logical observables, so they can be consolidated into a single node. Since each error in this set has the same action on the final logical observables, one can choose an arbitrary representative when checking for decoding success. Similarly, some errors will have no effect on the detectors or logical observables, and as such are not included as a node in the bipartite graph. This bipartite graph can be considered the Tanner graph of a classical code and can be decoded by any appropriate decoder to deduce the errors that have occurred. Figure~\ref{fig:detecting}(b) shows the bipartite decoding graph corresponding to the circuit of panel (a). The classes of equivalent errors from each detecting region constitute the bit nodes of the graph and are connected by edges to the appropriate detectors. 
For a more detailed discussion of the circuit-level noise decoding process, see Ref.~\cite{bravyi2023highthreshold}.

\subsection{Circuit-level simulations}
\label{sec:results}

We now present the results of circuit-level error correction simulations using the class of BB quantum LDPC codes and the architecture defined in Section~\ref{sec:arch}.
Previous simulations of BB codes showed that they greatly outperformed surface codes in terms of overhead under specific architecture assumptions~\cite{bravyi2023highthreshold, viszlai2023matching}.
Here, we show that BB codes implemented with 2D local gates in the proposed bilayer architecture have comparable performance to surface codes which encode the same number of logical qubits and have roughly the same number of physical qubits.

For the following simulations, we use Stim~\cite{gidney2021stim} to construct the circuits and build the space-time bipartite graph used for decoding.
As such, we consider a circuit-level noise model in which errors occur independently on different circuit operations. For a physical error rate $p$---in this work we consider $p=0.1\%$---single-qubit gates have probability $p/10$ of experiencing the single-qubit depolarizing channel;
two-qubit gates have probability $p$ of experiencing the two-qubit depolarizing channel;
measurement results have probability $p$ of being flipped;
qubit reset operations have probability $p/10$ of preparing the $\ket{1}$ state instead of the $\ket{0}$ state;
and idle qubits experience a depolarizing channel with probability $p/50$.
The assumed single-qubit, two-qubit, and measurement error rates are comparable to the performance of current ion-trap~\cite{moses2023, baldwin2022, ionq2024} and superconducting~\cite{ibm2021} quantum computers. However, this last condition on the idle qubit error rate is somewhat optimistic and is around an order of magnitude better than the idle error seen on production devices. We comment on this assumption in Section~\ref{sec:discussion}.

\begin{figure}[t]
    \centering
    \includegraphics[width=\linewidth]{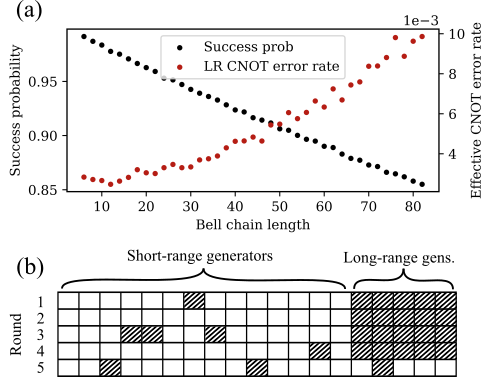}
    \caption{(a) Results of circuit-level simulations of the entanglement purification protocol of Ref.~\cite{Bennett_1996} for Bell pairs of increasing length. Two long-range Bell pairs are created using a noisy circuit similar to that of Fig.~\ref{fig:architecture}(a) and then purified with the noisy circuit depicted in Fig.~\ref{fig:purifcation}. The success probability of the purification and the resulting Bell purity if successful is shown for $100\,000$ samples. (b) Example depiction of generator masking (indicated by a hatched fill) over several error correction rounds being affected by the entanglement purification protocol failing. In this example, the long-range generators are unmasked after five rounds. }
    \label{fig:bell_purity}
\end{figure}

For ease of implementation, we first separately perform circuit-level simulations of the entanglement purification protocol. The simulation consists of implementing two noisy long-range Bell pairs using a circuit similar to that depicted in Fig.~\ref{fig:architecture}(a) and then performing Bennett \textit{et al.}'s entanglement purification protocol on the two pairs. 
In this simplest version of the protocol, failures are not reattempted, and only a single donor Bell pair is used.
Simulating the protocol many times allows us to estimate the probability that the purification protocol succeeds and, if so, the fidelity of the purified Bell pair.
Fig.~\ref{fig:bell_purity}(a) displays the results of these simulations for long-range Bell pairs of different lengths under the circuit-level error model described above. 

During syndrome extraction, if the entanglement purification protocol for any of the long-range CNOT gates fails, we mask the corresponding generator instead of reattempting the purifications. We can then estimate the probability that the syndrome of a generator is available, that is, all the required purifications for that generator succeed. 
If the purifications do succeed, then we can also estimate the error rate of the resulting long-range CNOT gate from the fidelity of the Bell pair. In the full circuit, we then implement a direct CNOT with this error rate to represent the entire procedure.
Fig.~\ref{fig:bell_purity}(b) illustrates what this means practically: assuming the long-range generators are unmasked every five rounds, the first four rounds have these long-range generators masked (hatched fill). 
Additionally, due to failures of the entanglement purification protocol, some short-range generators are also masked, even though we had planned for them to always be available. We note that these random failures are not expected to greatly impact the performance of the code, as it is unlikely that one generator will fail several rounds in a row. Thus, even if there are missed errors, they will likely be corrected when the generator does succeed in routing.
In the fifth round, the long-range generators are unmasked and attempted to be measured, but only if purifications succeed can we actually obtain their syndromes.
Note that with this simple purification scheme, the long-range generators are less likely to succeed, since the necessary Bell pairs are between more distant qubits and more prone to failure.

For the full error correction protocol, we begin each circuit with a single noiseless round to initialize the logical subspace. We then perform $t$ noisy error correction rounds using the syndrome extraction circuits defined in Section~\ref{sec:synd_ext_cir}. 
As the short-range generators are easier to measure, we attempt to measure them every round, whereas the costly, long-range generators are unmasked and attempted every five rounds.
As described above, we additionally mask both the short- and long-range generators with probability equal to that of at least one of required purifications failing. In the cases where all purifications for a single generator succeed, we apply the two-qubit depolarizing channel after each CNOT gate with an error rate equal to that of a long-range CNOT gate performed using a Bell pair of the appropriate distance.
Idling error rates are estimated using the number of steps needed to route and purify the source and donor Bell pairs for a given set of generators (see Fig.~\ref{fig:routing_schedule} and Table~\ref{tab:codes}). As each step consists of Bell pair generation (depth 6), purification (depth 2), and implementation of the long-range CNOT gate (depth 3), the actual circuit depth is $11\times$ greater. To represent idling errors, a depolarizing channel is applied at the beginning of each error correction round to every qubit with probability equal to the total circuit depth times the idle error rate.
Additionally, a depolarizing channel is applied to every qubit with probability $p=0.1\%$ at the beginning of each round.
Before measuring the logical observables, we noiselessly extract the full syndrome one last time. The corresponding space-time bipartite graph is then generated, and the errors are sampled and decoded.

In this work, we use a decoder based on belief propagation and ordered-statistics decoding (BP-OSD)~\cite{Panteleev_2021, Roffe_2020, Roffe_LDPC_Python_tools_2022}, which consists of the min-sum BP decoder followed by an order-10 combination-sweep OSD postprocessing step. 
Performing real time decoding using BP and higher-order OSD postprocessing may be infeasible within the fast cycle time of superconducting quantum computers; however, it was shown that good decoding performance for BB codes can be achieved while using less computationally expensive OSD parameters~\cite{scruby2024highthresholdlowoverheadsingleshotdecodable}.

\begin{table}[t]
\begin{tabular}{c|c|c}
$[[n,k,d]]$ & Qubits & $\epsilon_L$ \\
\hline
  $[[128,8,4]]$ & 248 &  $1.4\times 10^{-3}\pm 1.2\times 10^{-6}$  \\
  \hline
  $\mathbf{[[72,8,6]]}$     &  288      &    $1.6\times 10^{-3}\pm 3.0\times 10^{-5}$          \\
  \hline
  $\mathbf{[[90,8,6]]}$     &  360      &    $8.9\times 10^{-4}\pm 2.0\times 10^{-5}$          \\
  \hline
    $[[200,8,5]]$ & 392 &  $2.0\times 10^{-4}\pm 6.5\times 10^{-7}$ \\
  \hline
  $\mathbf{[[120,8,8]]}$    &  480      &    $1.2\times 10^{-4}\pm 2.0\times 10^{-6}$          \\
  \hline
  $[[288,8,6]]$ & 568 &  $9.5\times 10^{-5}\pm 2.5\times 10^{-7}$ \\
  \hline
  $\mathbf{[[150,8,8]]}$    &  600      &    $5.3\times 10^{-5}\pm 1.3\times 10^{-6}$          \\
  \hline
  $[[392,8,7]]$ & 776 &  $2.0\times 10^{-5}\pm 1.5\times 10^{-7}$ \\
  \hline
  \hline
  $\mathbf{[[144,12,12]]}$ & 576 & $1.6\times 10^{-4}\pm 4.6\times 10^{-6}$ \\
  \hline
  $[[300,12,5]]$ & 588 & $3.0\times 10^{-4}\pm 9.8\times 10^{-7}$ \\
  \hline
  $\mathbf{[[196,12,8]]}$ & 784 & $7.9\times 10^{-5}\pm 2.3\times 10^{-6}$ \\
  \hline
  $[[432,12,6]]$ & 852 & $1.4\times 10^{-4}\pm 3.7\times 10^{-7}$ \\
  \hline
  $[588,12,7]]$ & 1164 & $2.9\times 10^{-5}\pm 2.3\times 10^{-7}$ \\
  \hline
\end{tabular}
\caption{Code parameters, total number of qubits used, and $\epsilon_L$ as extracted from Eq.~\eqref{eq:logical_per_round} for the simulations described in Section~\ref{sec:results}. Code parameters shown in bold correspond to BB code instances. Code parameters not in bold correspond to copies of the rotated surface code. }
\label{table:code_results}
\end{table}

\begin{figure}[t]
    \centering
    \includegraphics[width=\linewidth]{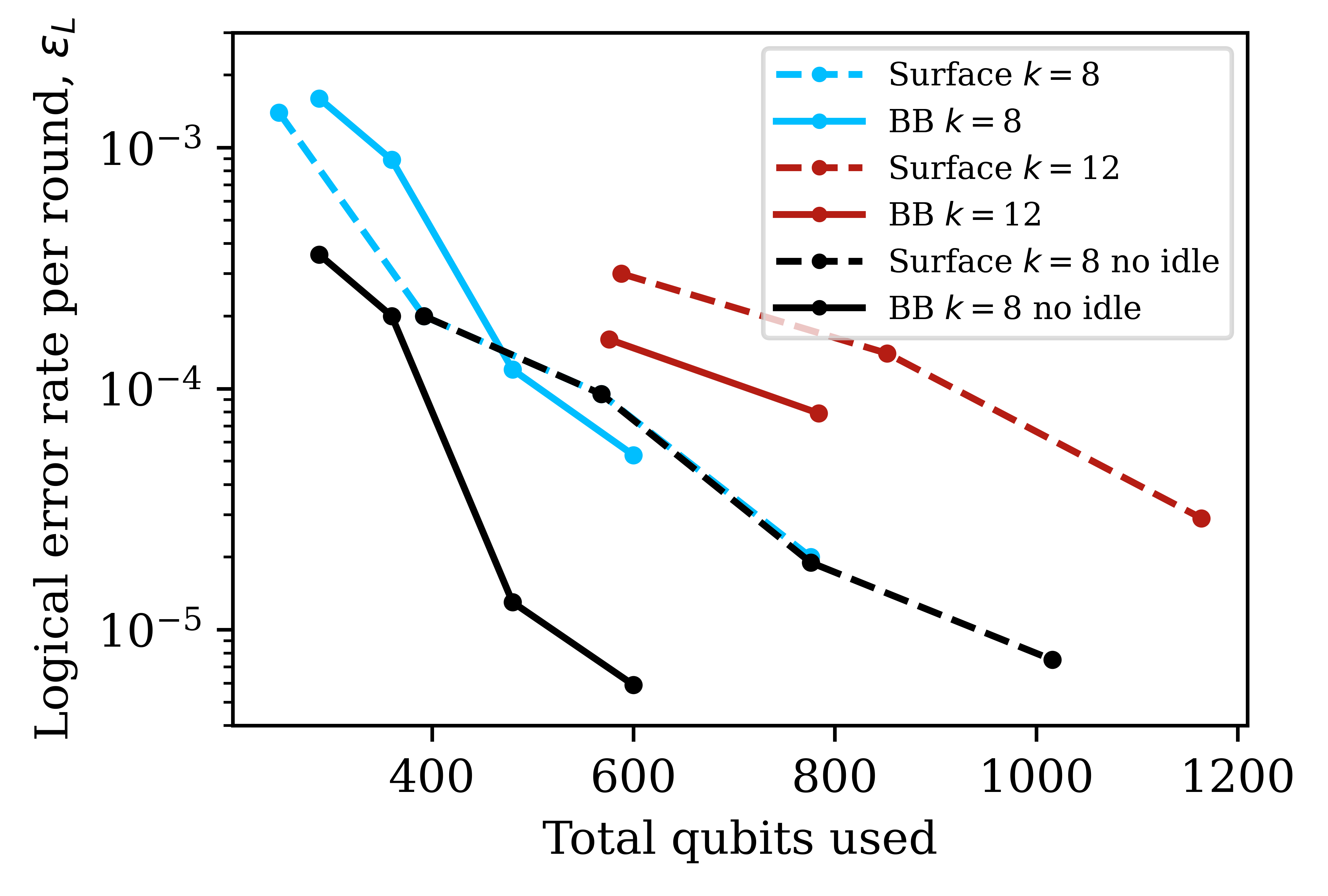}
    \caption{Extracted logical error rate per round, $\epsilon_L$ as a function of the total number of qubits used (data qubits plus all ancilla qubits) for several BB and surface code instances. The data is tabulated in Table~\ref{table:code_results}. We also include simulation results in the no idle error regime, as indicated by the black lines; these results are tabulated in Table~\ref{table:code_results2}.}
    \label{fig:space_vs_epsilon}
\end{figure}

Table~\ref{table:code_results} and Fig.~\ref{fig:circuit_level}(a)--(b) shows the results of these simulations for several codes listed in Table~\ref{tab:codes}.
As a comparison, we perform the same simulations with the rotated surface code which has parameters $[[d^2, 1, d]]$. 
To decode, we follow the same process as described in Section~\ref{sec:space_time} but instead use the minimum-weight perfect matching decoder~\cite{higgott2023sparse}.
As the BB codes encode multiple logical qubits in a single block, multiple copies of the surface code must be used to achieve the same number of logical qubits. If $p_{SC,1}$ is the logical error rate of simulating a single rotated surface code for $t$ error correction rounds, then $k$ copies of the surface code have a logical error rate
\begin{equation}
    p_{SC,k} = 1 - (1-p_{SC,1})^k.
    \label{eq:surface}
\end{equation}
In addition to the logical error rate, another important performance metric is the number of qubits used to achieve it. 
For the BB codes and the bilayer architecture, this includes the ancillary check qubits as well as the entire routing layer, which for an $[[n,k,d]]$ code uses $4n$ qubits in total.
The rotated surface code uses $d^2-1$ additional check qubits, which brings the total number of qubits to $2d^2-1$ for each copy. 
The total number of qubits used is listed together with the code parameters in Fig.~\ref{fig:circuit_level}.
The error bars on the data points are calculated using the standard error when sampling from a binomial distribution $\sqrt{p_{\log} (1-p_{\log} )/N}$, where $N$ is the number of collected samples. Due to the large number of shots taken, $N \sim 10^5$, error bars in Fig.~\ref{fig:circuit_level} and Fig.~\ref{fig:no-idle-results} are nearly invisible. Additionally, we plot a fit of 
\begin{equation}
    p_{\log} = 1 - (1 - \epsilon_L)^t
    \label{eq:logical_per_round}
\end{equation}
for both the surface and BB codes, from which we can extract the logical error rate per round, $\epsilon_L$.

The smallest BB codes encoding $k=8$ logical qubits are outperformed by surface codes which use fewer physical qubits. However, increasing the block length yields BB codes that
surpass the performance of similarly sized surface codes. This is illustrated in Fig.~\ref{fig:space_vs_epsilon}, where we see the BB codes achieving a lower logical error rate per round than the surface code while utilizing fewer qubits. Increasing the number of logical qubits to $k=12$, BB codes and the proposed architecture immediately outperform the surface codes in terms of logical error rate and space overhead.
Compared to twelve patches of a $[[36,1,6]]$ rotated surface code using a total of 852 physical qubits and a logical error rate per round of $\epsilon_L = 1.43\times 10^{-4}$, we find a $[[144,12,12]]$ BB code using 576 qubits that matches the performance, with $\epsilon_L = 1.56\times 10^{-4}$. Additionally, we find a $[[196,12,8]]$ code using 784 qubits that outperforms it with $\epsilon_L = 7.89\times 10^{-5}$.
At this scale, the improvements are not so drastic, but we expect to see greater overhead benefits as the block length and number of logical qubits increase.

We now vary the interval at which the long-range generators are measured, the results of which are also shown in Fig.~\ref{fig:circuit_level}(c). 
For the $[[90,8,6]]$ code there are 44 (not necessarily independent) generators of a single type. Using a routing schedule that was found by hand, all 35 short-range generators of a single type can be routed, purified, and measured in 9 steps; whereas it takes 5 steps to route, purify, and measure the remaining 10 long-range generators of the same type. Measuring the 35 short-range and 10 long-range generators of the opposite type requires an additional 9 and 5 steps, respectively. This means that measuring the long-range generators every five error correction rounds requires a circuit depth that is 28.5\% shorter than if the long-range generators were measured every round ($4\cdot 2\cdot 9 + 2\cdot(5+9)=100$ steps versus $5\cdot 2\cdot (5+9) = 140$ steps), at a negligible increase in the logical error rate per round from $\epsilon_L = 8.41\times 10^{-4}$ to $8.92\times 10^{-4}$. Increasing the size of the code to $[[196,12,8]]$, we again see negligible differences in logical error per round performance between the two measurement schedules, from $\epsilon_L = 7.41\times 10^{-5}$ to $7.89\times 10^{-5}$, with the additional benefit of a 32.0\% decrease in the depth of the syndrome extraction circuit when measured every five rounds.
The two codes consist of $22.22\%$ and $35.71\%$ long-range generators, respectively, yet both remain consistent between long-range measurement intervals.
As the block length increases, so does the discrepancy between the measurement times of the short- and long-range generators, increasing the circuit depth savings.
Additionally, this discrepancy would disproportionally introduce more errors during the long-range measurement rounds, potentially making it more performant to measure these large generators even less frequently. However, this behavior is very dependent on the idle error rate, and changing the idle error rate may cause the two curves to deviate.

\begin{figure}[t]
    \centering
    \includegraphics[width=0.95\linewidth]{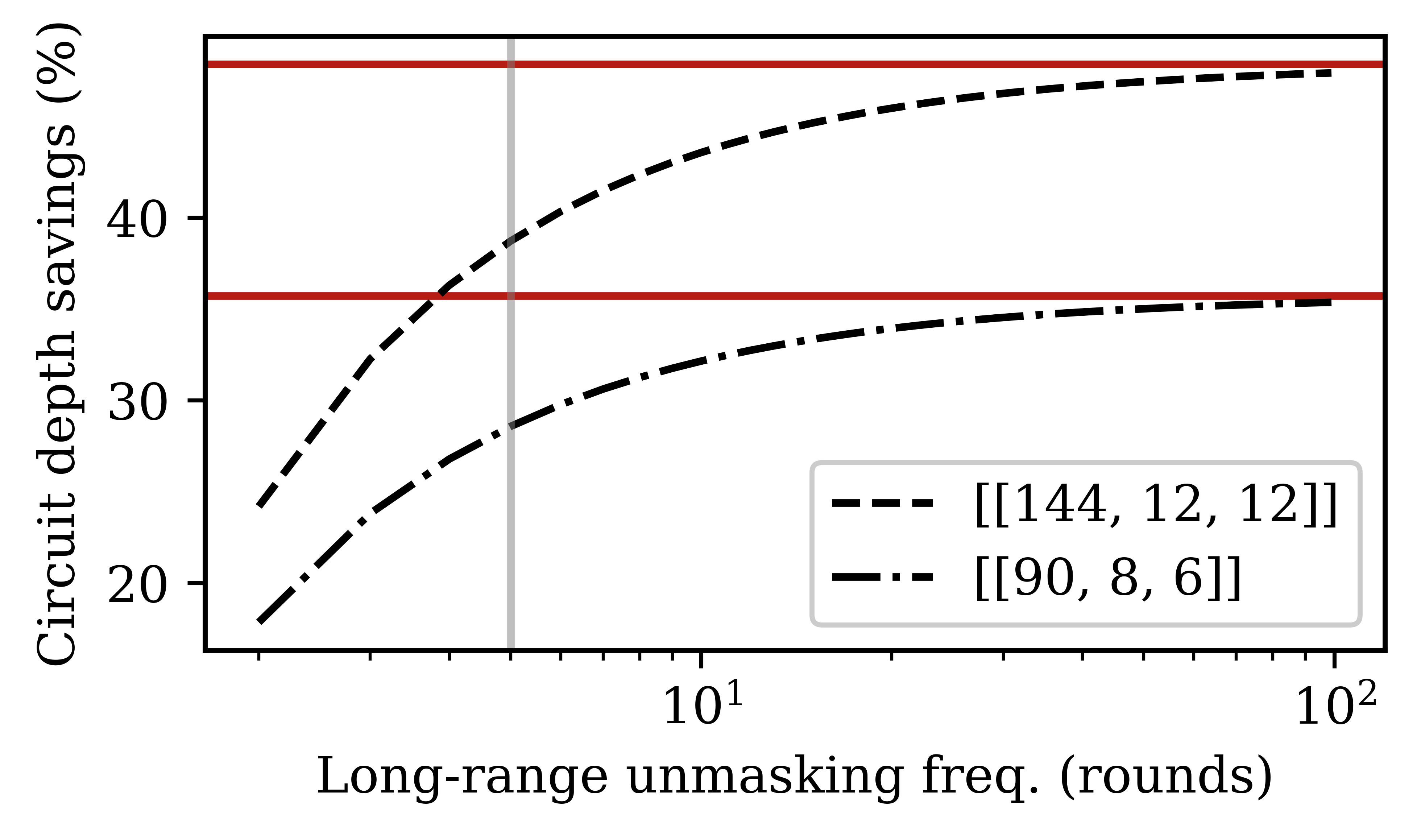}
    \caption{Percentage change in circuit depth, compared to a circuit that always measures every generator, as a function of number of rounds elapsed between long-range generator measurements. Horizontal solid red lines indicate the potential maximum reduction in circuit depth for the two BB code instances. The gray vertical line highlights the depth savings achieved by measuring the long-range generators every five error correction rounds, as done in Fig.~\ref{fig:circuit_level}(a)--(b) and Fig.~\ref{fig:no-idle-results}.}
    \label{fig:pot_savings}
\end{figure}

Achieving more significant reductions in circuit depth requires measuring the long-range generators much less frequently, as shown in Fig.~\ref{fig:pot_savings}. We display the potential circuit depth savings for two BB code instances as a function of how many error correction rounds elapse between measurements of the long-range generators.
The horizontal red lines indicate the maximum potential savings, corresponding to a schedule where the long-range generators are only measured once at the end of the circuit. 
For example, the $[[144,12,12]]$ BB code requires 16 (15) steps to measure the short- (long)-range generators of a single type, which gives a maximum depth savings of 48.4\%. 
When measuring every five rounds, as in Fig.~\ref{fig:circuit_level}(a)--(b) and Fig.~\ref{fig:no-idle-results}, we see a circuit depth savings of 38.7\%, indicated by the vertical gray line.
Measuring the long-range generators very infrequently will significantly degrade the error correction performance and may not be worth the reduced circuit depth. Instead, it may be more advantageous to measure the long-range generators relatively frequently, e.g., every 2--5 rounds; in that regime we still see considerable circuit depth savings (50\% to 80\% of the theoretical maximum), but the impact on the logical error rate is negligible.

\begin{table}[]
\begin{tabular}{c|c|c}
$[[n,k,d]]$ & Qubits & $\epsilon_L$ \\
\hline
  $\mathbf{[[72,8,6]]}$     &  288      &    $3.6\times 10^{-4}\pm 4.9\times 10^{-6}$          \\
  \hline
  $\mathbf{[[90,8,6]]}$     &  360      &    $2.0\times 10^{-4}\pm 3.3\times 10^{-6}$          \\
  \hline
    $[[200,8,5]]$ & 392 &  $2.0\times 10^{-4}\pm 6.5\times 10^{-7}$ \\
  \hline
  $\mathbf{[[120,8,8]]}$    &  480      &    $1.3\times 10^{-5}\pm 8.2\times 10^{-7}$          \\
  \hline
  $[[288,8,6]]$ & 568 &  $9.5\times 10^{-5}\pm 2.5\times 10^{-7}$ \\
  \hline
  $\mathbf{[[150,8,8]]}$    &  600      &    $5.9\times 10^{-6}\pm 2.4\times 10^{-7}$          \\
  \hline
  $[[392,8,7]]$ & 776 &  $2.0\times 10^{-5}\pm 1.5\times 10^{-7}$ \\
  \hline
  $[[512,8,8]]$ & 1016 & $7.5\times 10^{-6}\pm 4.3\times 10^{-8}$
\end{tabular}
\caption{Code parameters, total number of qubits used, and $\epsilon_L$ as extracted from Eq.~\eqref{eq:logical_per_round} for the simulations described in Section~\ref{sec:results} albeit with no idle error. Code parameters shown in bold correspond to BB code instances. Code parameters not in bold correspond to copies of the rotated surface code. }
\label{table:code_results2}
\end{table}

Even with the reduced idle error rate that we consider here, idle errors are a significant error source, especially on rounds where the long-range generators are measured. 
In Table~\ref{table:code_results2} and Fig.~\ref{fig:no-idle-results}, we perform the same simulations as described above but do not apply idling errors. Due to their short syndrome extraction circuit depths, the surface codes are unaffected from the decrease in idle error. However, we now find that all BB code instances achieve better logical error rates than surface codes while using fewer physical qubits, as shown in Fig.~\ref{fig:space_vs_epsilon}. Indeed, the $[[150,8,8]]$ code using 600 physical qubits sees an $8.8\times$ improvement in the logical error rate per round, from $\epsilon_L = 5.31\times 10^{-5}$ to $5.9\times 10^{-6}$, and now outperforms eight patches of a $[[64,1,8]]$ rotated surface code using $1016$ physical qubits with $\epsilon_L = 7.5\times 10^{-6}.$ Achieving negligible idle error rates may not be feasible, but it illustrates the regime where our protocol performs best.

\section{Discussion}
\label{sec:discussion}

In this paper, we have presented a bilayer architecture for implementing nonlocal qLDPC codes on quantum devices which are restricted to 2D local gates. We have shown that bivariate bicycle codes are well suited for such an architecture and described a parallelizable syndrome measurement scheme which makes use of the geometric parity check structure of the codes. Through circuit-level simulations of a multi-round decoding protocol, we found that BB codes attain comparable logical error rates to that of the rotated surface code while using fewer physical qubits. Furthermore, by applying the stacked model and masking, we achieved a significant decrease in syndrome extraction time with negligible impact on the error correction performance.

However, there are a number of challenges that must be considered in a physical implementation of this protocol.
Perhaps the most notable issue is the depth of the circuit required to perform even a single syndrome extraction. 
Implementing a single long-range CNOT gate requires constructing the long-range Bell pair, purifying it, and using it to implement a CNOT between a data qubit and check qubit.
Although several CNOT gates can be implemented in parallel, doing this for the entire set of generators requires 10s of routing steps, translating to a physical circuit with depth in the 100s.  One consequence of the depth of the circuit is that our protocol only performs well in the regime of low idle error rate.
Furthermore, per Claim~\ref{claim:dist_routing}, as the block length increases so too does the required routing time and, consequently, the physical circuit depth.
This is in stark contrast to the implementation in Ref.~\cite{bravyi2023highthreshold}, where the entire set of generators can be measured with a circuit of depth seven, albeit with the use of long-range connections. These long-range connections are a significant engineering challenge, and it is unclear whether implementing high-fidelity gates in this way is feasible.

We do find BB codes where the same parity check structure is shared between codes of increasing block length. For code families with this property, the number of steps in the syndrome extraction circuit is constant, so the noise per syndrome extraction cycle coming from idle error does not increase.  However, this also means that the percentage of long-range generators and, by extension, the amount of nonlocality in the code, decreases.
References \cite{Baspin_2022, Baspin_2022_2} showed that it is impossible to beat the asymptotic scaling of the surface code parameters without increasing the amount of nonlocality. Increasing the block length will yield larger $k$ and $d$, but the asymptotic scaling of these codes will approach that of the surface code; however, for finite sizes we would still expect to see significant space overhead savings compared to alternative, lower-rate codes. Even for BB codes with increasing generator shape, it is feasible that the increased error correction capabilities will outpace the increase in idle error. In particular, Fig.~\ref{fig:routing-algo-results} shows a power-law relationship between block length and routing depth. Assuming the code is operating below threshold, the exponential suppression in logical error should be sufficient to handle the increased effective idle error. 

Another challenge is that the simple purification protocol presented here does not scale well, as increasing the block length would lead to low-fidelity Bell pairs and a high purification failure rate. 
Although there are many entanglement purification protocols that improve the resulting Bell fidelity~\cite{Jiang_2007,dur2007entanglement,krastanov2019optimized,gidney2023tetrationally}, using them would further increase the depth of the syndrome extraction circuits or require additional ancillary qubits.
The one potential saving factor is that the vast majority of the work is done by the upper routing layer to construct and purify the Bell pairs, and the two layers interact in fewer than $1/10$ of the circuit steps. If it were possible to sufficiently isolate the data layer, akin to what is done in ion traps or reconfigurable atom arrays, it might be possible to achieve the low idling error rates that would greatly improve the performance of the protocol.

If the aforementioned issues can be solved, then scaling up should increase the advantage of qLDPC codes over the surface code.
One potential solution is to improve the circuit depth of the protocol. An architectural feature that could accomplish this is the ability to perform two-qubit gates on qubits that are some distance $R$ apart~\cite{pattison2023hierarchical}.
This is a natural operation on neutral-atom devices, where Rydberg-Rydberg interactions, especially dipolar ones~\cite{saffman10a}, can be quite long-range.
Furthermore, Rydberg-Rydberg interactions can help with syndrome extraction by naturally realizing long-range generalized multi-control multi-target CNOT gates~\cite{young2021}. At the same time, such long-range Rydberg-based gates may harm parallelism since only one such gate can be implemented within the Rydberg blockade radius at a time.
Gates beyond nearest-neighbor could also be feasible in superconducting devices through the use of medium-range couplers or photonic interconnects~\cite{leung2018deterministic}. When $R$ is a constant, the asymptotic behavior will remain unchanged; however, practically this would mean that the short-range generators would be much easier to implement. With an appropriate choice of $R$, it would then be possible to use the depth-7 circuit of Ref.~\cite{bravyi2023highthreshold} to measure the short-range generators, in which case the only difficulty would be to measure the long-range generators in the proposed manner. 
An alternative approach would be to add additional ancilla layers to the architecture. Although this would further increase the qubit overhead, it would allow for more parallelization during the syndrome measurement, decreasing the total circuit depth.

\section*{Acknowledgements}
We thank Patrick Rall for answering questions about Ref.~\cite{bravyi2023highthreshold}. A.M.C., A.V.G., M.J.G, and D.G.\ were supported in part by the National Science Foundation (QLCI grant OMA-2120575). A.M.C. and A.V.G. were supported in part by the DoE ASCR Accelerated Research in Quantum Computing program (award No.~DE-SC0020312). A.M.C., A.V.G., and D.D. were supported in part by the DoE ASCR Quantum Testbed Pathfinder program (awards DE-SC0019040 and DE-SC0024220). A.V.G.~was also supported in part by the NSF STAQ program, AFOSR, AFOSR MURI, and DARPA SAVaNT ADVENT. Support is also acknowledged from the U.S.~Department of Energy, Office of Science, National Quantum Information Science Research Centers, Quantum Systems Accelerator.  D.D.\ acknowledges support by the NSF GRFP under Grant No.~DGE-1840340 and an LPS Quantum Graduate Fellowship. E.S.~was supported by the U.S. Department of Energy, Office of Science, National Quantum Information Science Research Centers, Quantum Science Center. 

\section*{Data Availability}
The source code and data to generate the figures in the paper are available at \url{https://github.com/noahberthusen/qecc_routing}. 

\bibliography{bibliography}

\begin{widetext}
    
\end{widetext}

\appendix

\section{Bivariate bicycle code embeddings}
\label{apx:embedding}

Here we briefly describe the conditions for embedding bivariate bicycle codes in a 2D grid. See Ref.~\cite{bravyi2023highthreshold} for a more complete discussion.

\begin{definition}[{\cite[Definition 1]{bravyi2023highthreshold}}]
    \label{def:toric}
    A code $QC(A,B)$ has a toric layout if its Tanner graph has a spanning sub-graph isomorphic to the Cayley graph of $\mathbb{Z}_{2\mu}\times \mathbb{Z}_{2\lambda}$ for some integers $\mu, \lambda$.
\end{definition}

This is to say that codes with a toric layout have checks that act on the four nearest-neighbor qubits, and potentially on additional nonlocal qubits. The four nearest-neighbor qubits can be measured using a standard surface code syndrome extraction circuit~\cite{Tomita_2014}, whereas the nonlocal qubits are measured using the proposed protocol. In the following, the order of an element $\ord(M)$ of a multiplicative matrix group is the smallest positive integer such that $M^{\ord(M)}=\mathbb{I}$, where $\mathbb{I}$ is the identity matrix of the same dimension as $M$.

A BB code $QC(A,B)$ depends on choices of matrices $A$ and $B$, as in Eq.~\eqref{eq:ab}, whose terms are powers of $x$ or $y$, defined in Eq.~\eqref{eq:xy}. The matrices $x$ and $y$ depend on choices of positive integers $\ell,m$, and they correspond to the dimensions of the grid in which the code $QC(A,B)$ is embedded should it satisfy Lemma~\ref{lem:toric_layout}. The $\mu$ and $\lambda$ of Definition~\ref{def:toric} are $\ell$ and $m$, respectively.
In this toric layout, qubits and checks can be labeled by 
$\mathcal{M}$, which can be considered to be a list of integers $\mathbb{Z}_{\ell m}=\{0,1,\ldots,\ell m-1\}$ that represent locations on the 2D grid.

\begin{lemma}[{\cite[Lemma 4]{bravyi2023highthreshold}}]%
    \label{lem:toric_layout}
    A code $QC(A,B)$ has a toric layout on a $2\ell \times 2m$ grid if there exist $i,j,g,h\in\{1,2,3\}$ such that
    \begin{enumerate}
        \item $\langle A_i A_j^T, B_g B_h^T\rangle = \mathcal{M}$
        \item $\ord(A_i A_j^T)  \ord(B_g B_h^T) = \ell m$
    \end{enumerate}
\end{lemma}

Here, $\langle A_i A_j^T, B_g B_h^T\rangle$ indicates the group generated by $A_i A_j^T$ and $B_g B_h^T$. The matrices $A_iA_j^T$ and $B_gB_h^T$ then correspond to horizontal and vertical translations, respectively, on the grid. To have a toric layout, these translations must visit the $\ell m$ $X$- and $Z$-type checks, as well as the two sets of $\ell m$ data qubits. 
Practically, this can be checked by multiplying $(B_gB_h^T)^b(A_iA_j^T)^a$ for $0 \le b < \ord(B_gB_h^T)$, $0 \le a < \ord(A_iA_j^T)$ with a basis vector of $\mathbb{F}_2^{\ell m}$ and seeing whether the other $\ell m - 1$ basis vectors can be obtained.
Satisfying this is equivalent to satisfying condition 1.
For a given choice of $A = A_1 + A_2 + A_3$ and $B = B_1 + B_2 + B_3$, there might not be assignments of $i,j,g,h$ such that Lemma~\ref{lem:toric_layout} is satisfied. There may also be several satisfying assignments. 
Each satisfying assignment yields an embedding with a defined generator shape, which in turn determines the fraction of generators that cross the long boundary condition. 
Thus, the embedding controls the routing schedule and number of masked generators, both of which affect the overall error correction performance of the code.

\clearpage

\begin{figure}[t]
    \centering
    \includegraphics[width=0.9\linewidth]{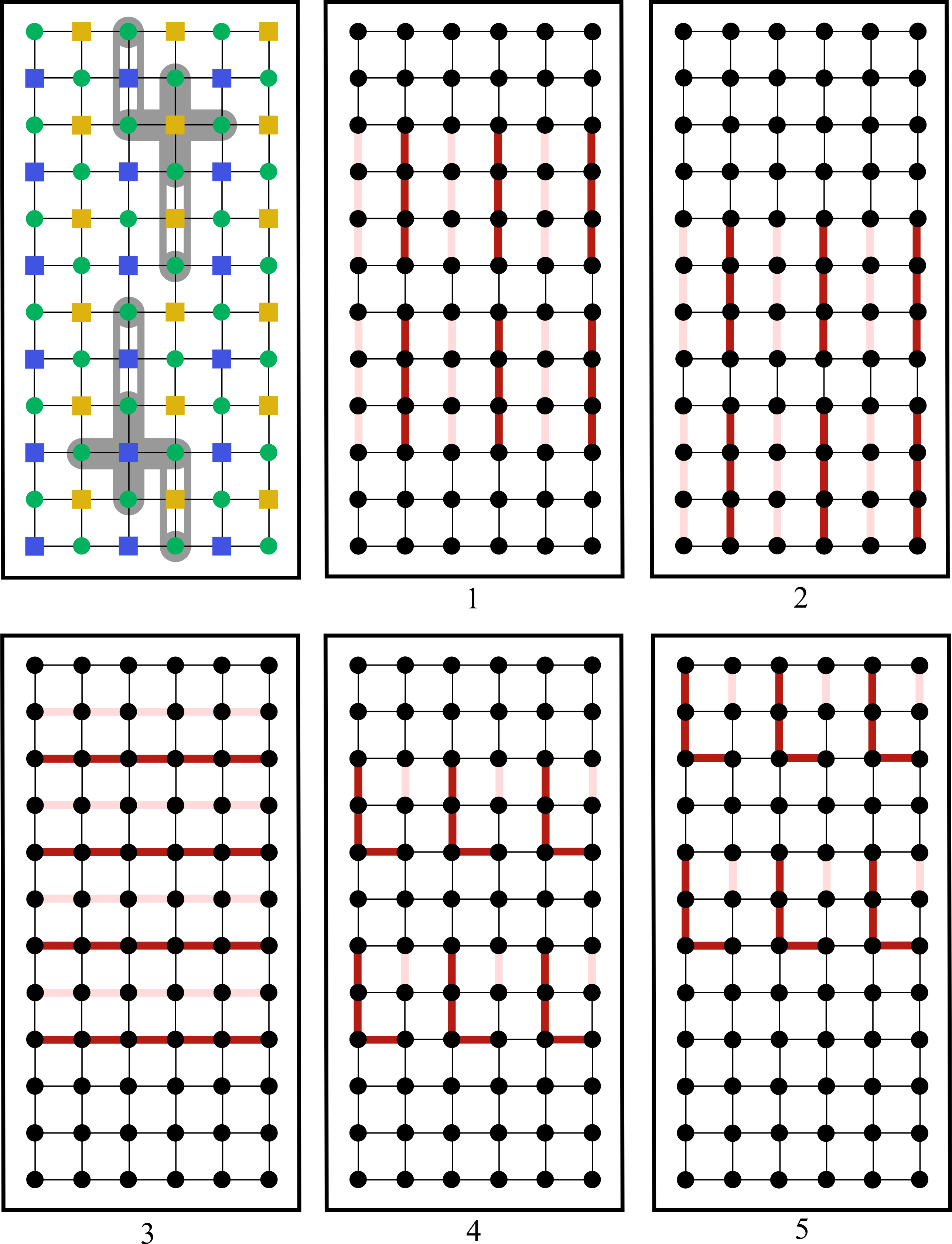}
    \caption{Example five-step schedule to route and purify the Bell pairs needed to measure the short-range, $Z$-type generators of a $[[36,4,4]]$ BB code constructed with $\ell=6, m=3$ and by matrices $A = x+y^3+y^2, B=y^3+x^5+x^4$. The first panel shows the structure of the $Z$-type checks (yellow squares) and $X$-type checks (blue squares), as outlined in gray. }
    \label{fig:routing_schedule}
\end{figure}

\begin{figure}[t]
    \centering
    \includegraphics[width=0.54\linewidth]{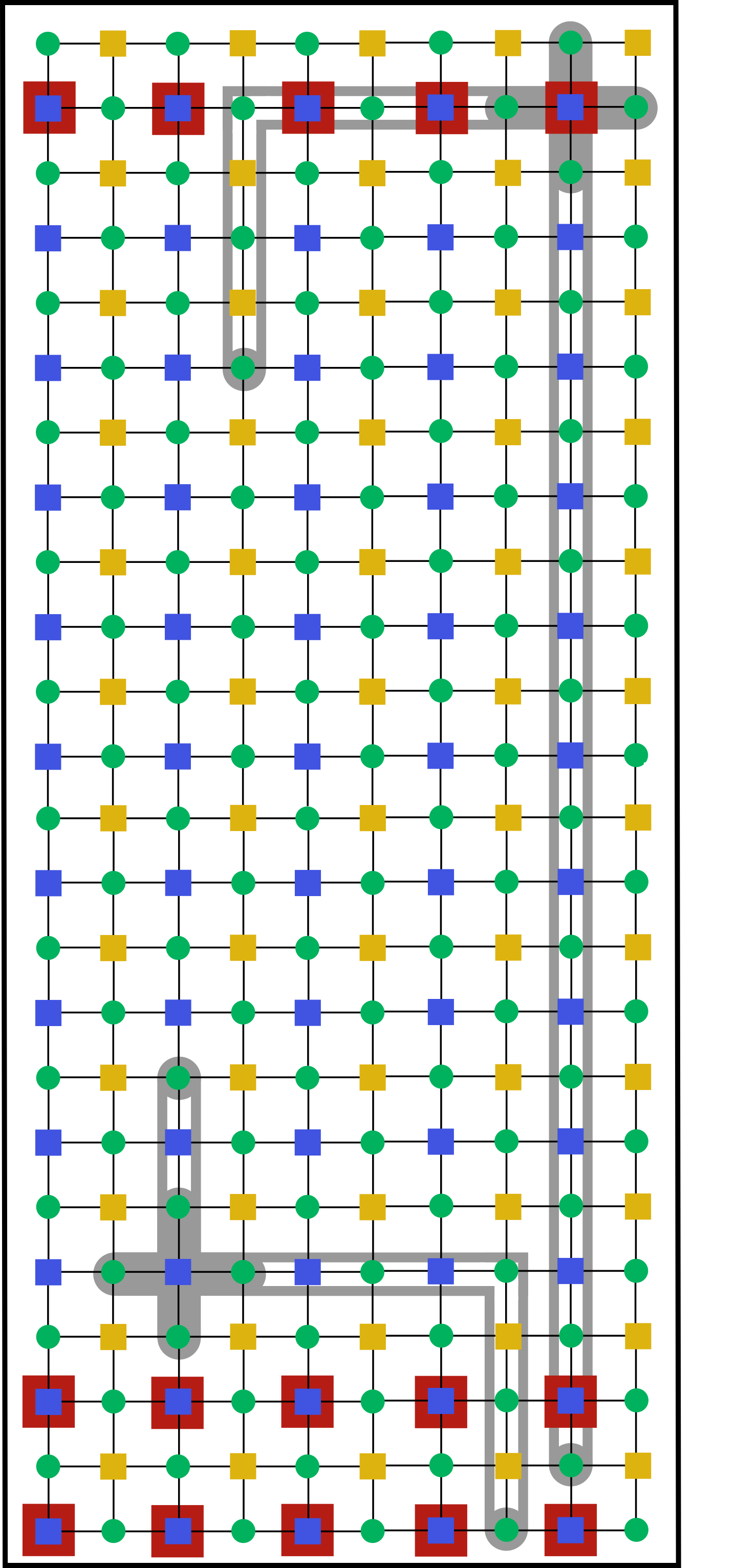}
    \caption{Illustration of the long- and short-range generators of a single type for the $[[120,8,8]]$ BB code constructed with $\ell=12, m=5$ and by matrices $A = x^{10}+y^4+y, B=1+x+x^2$. The qubits contained in the check are highlighted in gray. The checks (blue squares) highlighted in red are the long-range checks, as they cross the long boundary. For both $Z$- and $X$-type checks, there are 12 long-range checks out of a total of 48, yielding a mask percent of 25\%.}
    \label{fig:lrsr}
\end{figure}
\clearpage

\begin{figure*}[t]
    \centering
    \includegraphics[width=\textwidth]{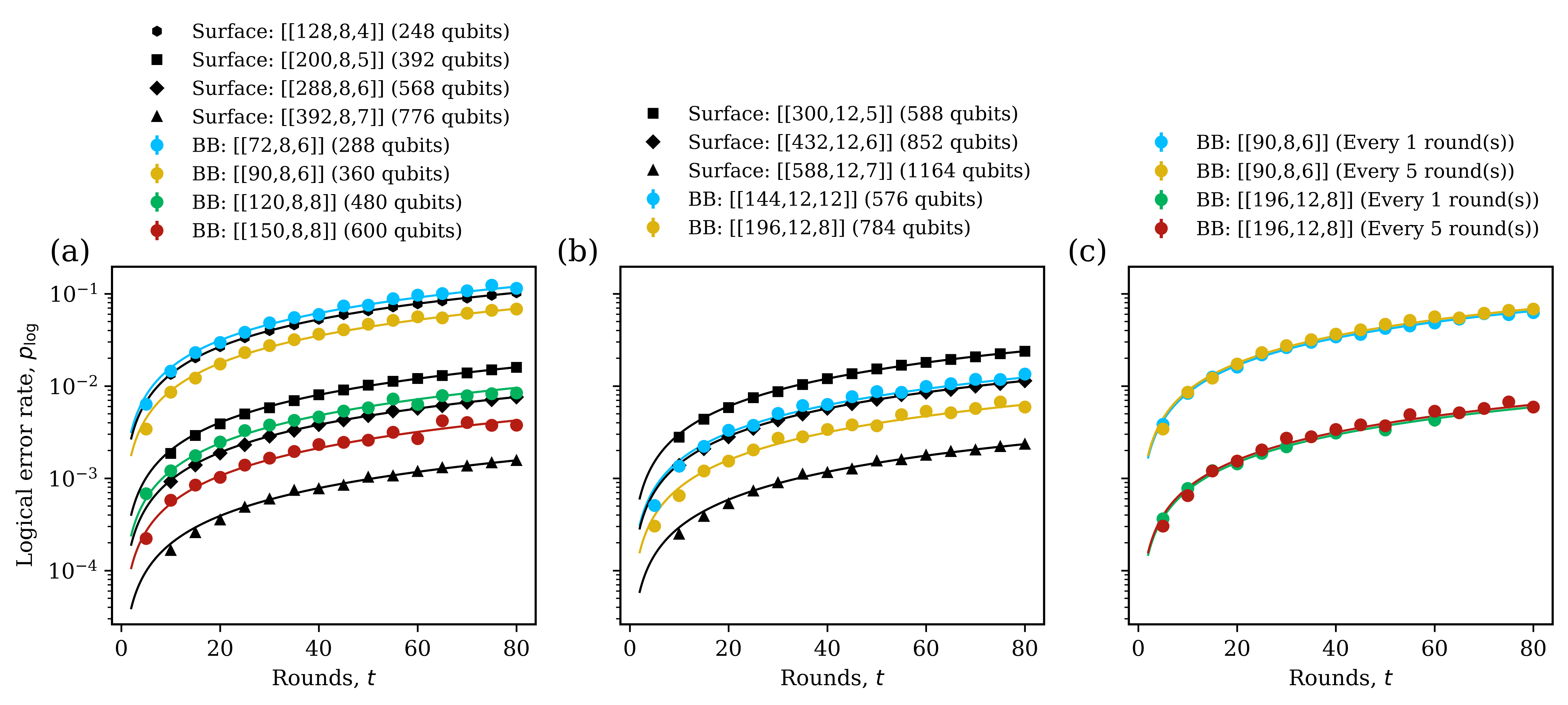}
    \caption{(a)--(b) Logical error rate of performing $t$ rounds of error correction with BB codes with (a) $k=8$ and (b) $k=12$ on the proposed bilayer architecture. The logical error rate of the same simulations using $k$ copies of rotated surface code is plotted as a comparison. (c) Comparison between time intervals at which the long-range generators are measured. For the BB codes in panels (a)--(b), the long-range generators were measured every five error correction rounds. A fit of Eq.~\eqref{eq:logical_per_round} is also shown, from which we extract the logical error rate per round, $\epsilon_L$.}
    \label{fig:circuit_level}
\end{figure*} 

\begin{figure}[t]
    \centering
    \includegraphics[width=0.95\linewidth]{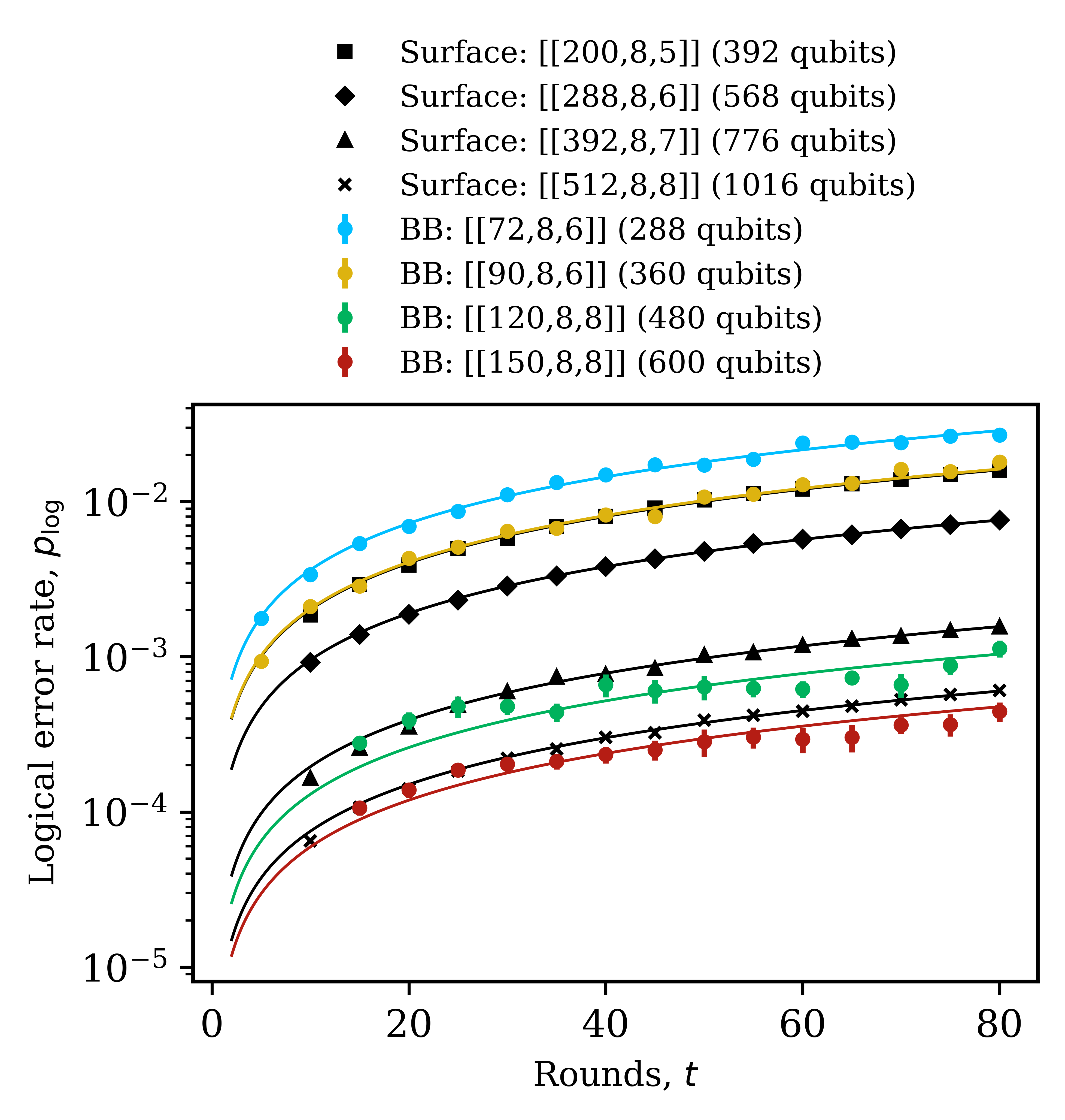}
    \caption{Logical error rate of performing $t$ rounds of error correction with BB codes with $k=8$ logical qubits on the proposed bilayer architecture. For this plot, we consider the case where the idle error rate is zero. The logical error rate of $k$ copies of the rotated surface code, calculated using Eq.~\eqref{eq:surface}, as well as the total number of physical qubits used, is again plotted as a comparison. 
    A fit of Eq.~\eqref{eq:logical_per_round} is also shown, from which we extract the logical error rate per round, $\epsilon_L$.
    } 
    \label{fig:no-idle-results}
\end{figure}

\end{document}